\documentclass[11pt]{article}
\usepackage{amsmath,amssymb,bbm,amsthm}
\usepackage{fullpage}
\usepackage{thm-restate,color,xcolor,xspace}
\usepackage{hyperref,cleveref}
\usepackage{graphicx}
\usepackage{algorithm,algorithmic}
\usepackage{enumitem}
\usepackage{crossreftools}
\usepackage{authblk}

\newtheorem{theorem}{Theorem}[section]
\newtheorem{lemma}[theorem]{Lemma}
\newtheorem{definition}[theorem]{Definition}
\newtheorem{corollary}[theorem]{Corollary}
\newtheorem{observation}[theorem]{Observation}
\newtheorem{claim}[theorem]{Claim}
\newtheorem{subclaim}[theorem]{Subclaim}

\newenvironment{subproof}[1][\proofname]{
	
	\begin{proof}[#1]
	}{
	\end{proof}
}

\begin{document}

	\title{Bellman-Ford in Almost-Linear Time}
	\author[1]{Isaac M. Hair}
	\author[2]{George Z. Li}
	\author[2]{Jason Li}
	\author[2]{Junkai Zhang}
	\affil[1]{UCSB, UCLA    \texttt{isaacmhair@gmail.com}}
	\affil[2]{Carnegie Mellon University $\{\texttt{gzli}, \texttt{jmli},\texttt{junkaizh}\}$\texttt{@cs.cmu.edu}}
	\date{}
	\maketitle
	
	\begin{abstract}
		We consider the single-source shortest paths problem on a directed graph with real-valued (possibly negative) edge weights and solve this problem in $m^{1+o(1)}$ time.
	\end{abstract}

	\section{Introduction}

	We consider the problem of computing single-source shortest paths on weighted directed graphs. When all edge weights are non-negative, the classic Dijkstra's algorithm solves the problem in $O(m+n\log n)$ time, which has been improved for sparse graphs to $O(m\sqrt{\log n\log\log n})$~\cite{duan2025breaking,duan2026faster}. When edge weights are integral (but possibly negative), the problem was solved in $\tilde O(m)$ time\footnote{$\tilde O(\cdot)$ ignores factors poly-logarithmic in $n$.}~\cite{bernstein2025negative} and $m^{1+o(1)}$ time~\cite{chen2025maximum}, with the latter solving the more general problem of minimum-cost flow. In the general case of real-valued edge weights, the 70-year-old Bellman-Ford algorithm was the fastest until the recent breakthrough $\tilde{O}(mn^{8/9})$ time algorithm of Fineman~\cite{fineman2024single}. Subsequent work has refined Fineman's approach~\cite{huang2025faster,huang2026faster,quanrud2025sparsification}, improving the runtime to $\tilde{O}(mn^{0.696}+(mn)^{0.850})$. 
	
	Li, Li, Rao, and Zhang \cite{li2025shortcutting} introduced a new approach for the problem. Their algorithm adds shortcut edges to the graph to reduce the number of negative edges (a.k.a.\ hops) along shortest paths by a constant factor. By iterating this procedure, they can compute the shortest paths in the shortcutted graph using an $O(1)$-hop shortest path algorithm. Unfortunately, this approach increases the number of edges: even in the first iteration, the algorithm may add $O(m\sqrt{n})$ edges to the graph. Further, the number of edges at least doubles per iteration, so the graph becomes dense after a few iterations of shortcutting. Subsequent work by \cite{li2026bellmanfordalmostlineartimedense} and \cite{khanna2026n2o1timealgorithmsinglesource} independently refined this approach to add $n^{2+o(1)}$ shortcut edges, obtaining an $n^{2+o(1)}$-time algorithm.
	
	In our work, we give a significantly improved analysis of a subroutine called \emph{betweenness reduction}, which implies that the number of edges added in each iteration is only $m^{1+o(1)}$ instead of $O(m\sqrt{n})$. We then give an \emph{unfolding lemma} to show that the number of edges added in each iteration can be bounded based on the number of original edges in the graph, even after adding shortcut edges. This prevents the number of edges from blowing up throughout the iterative shortcutting process. Combining these two tools, we get the following result.

	\begin{theorem}\label{thm:main}
		There is a (Las Vegas) randomized algorithm for single-source shortest paths on real-weighted directed graphs which runs in $m^{1+o(1)}$ time, with high probability.\footnote{With high probability means with probability at least $1-1/n^C$ for arbitrarily large constant $C>0$.}
	\end{theorem}

	\subsection{Graph Notation}
	
	All graphs $G=(V,E)$ in this paper are directed and have (possibly negative) real edge weights. If a vertex $u\in V$ has negative outgoing edges, we say $u$ is a \emph{negative vertex}, and let $N$ denote the set of negative vertices. For a vertex $u\in V$, we define $N^{\textup{out}}(u)=\{v\in V:(u,v)\in E\}$ and $N^{\textup{in}}(u)=\{v\in V:(v,u)\in E\}$ as the out-neighbors and in-neighbors of $u$. Finally, we let the distance between a set and a point be the minimum distance, i.e., $d(S,v)=\min_{s\in S}d(s,v)$.
	
	Like prior work, we heavily rely on a subroutine for efficiently computing single-source shortest paths with few negative edges. For a nonnegative integer $h$ called the \emph{hop bound}, define the \emph{$h$-negative hop distance} $d^h(s,t)$ as the minimum weight $(s,t)$-path with at most $h$ negative edges. Given a source vertex $s$, we can compute $d^h(s,t)$ for all other $t\in V$ in $O(h(m+n\log n))$ using a hybrid of Dijkstra and Bellman-Ford~\cite{dinitz2017hybrid}.
	
	Also like prior work, our algorithm makes use of \emph{potential functions} $\phi:V\to\mathbb R$ that reweight the graph, where each edge $(u,v)$ has new weight $w_\phi(u,v)=w(u,v)+\phi(u)-\phi(v)$. The key property of potential functions is that for any vertices $s,t\in V$, all $(s,t)$-paths have their weight shifted by the same additive $\phi(s)-\phi(t)$. A potential $\phi$ is \emph{valid} if $w_\phi(u,v)\ge0$ for all edges $(u,v)$ with $w(u,v)\ge0$, i.e., it does not introduce any new negative edges.

    During an iteration of our algorithm and those in prior work, we will compute a valid reweighting $\phi$, which may cause a negative edge in $G$ to become non-negative. However, even if a negative edge in $G$ is non-negative after the reweighting $\phi$, we will still call this edge negative through the rest of the iteration and count it towards hops in $h$-hop shortest paths. Following prior work, we refer to this as \emph{freezing} the set of negative edges before applying a reweighting.
	
	Finally, we will preprocess the graph using the following routine, so that each vertex $u$ has at most one negative edge incident to $u$. This preprocessing step is implicit in \cite{fineman2024single}, and our construction is essentially the same as the one given in Lemma 4 of \cite{li2025shortcutting}. 
	\begin{lemma}\label[lemma]{lem:one-negative-outgoing-edge}
		Given a graph $G=(V,E)$ with $k$ negative vertices, there is a linear-time algorithm that outputs a graph $G'=(V',E')$ with $V'\supseteq V$ such that
		\begin{enumerate}
			\item For any $s,t\in V$ and hop bound $h$, we have $d_{G'}^h(s,t)=d_G^h(s,t)$.
			\item Every vertex is adjacent to at most one negative edge. If a vertex $u$ is the tail of a negative edge, then that edge is the only outgoing edge from $u$. In addition, if a vertex $u$ is the head of a negative edge, then that edge is the only incoming edge to $u$.
			\item $G'$ also has $k$ negative vertices, and $|V'|=|V|+2k$.
		\end{enumerate}
	\end{lemma}
	\begin{proof}
		For each negative vertex $u$, consider its negative outgoing edges $(u,v_1),\ldots,(u,v_\ell)$, and suppose that $(u,v_1)$ is the edge of smallest weight, with ties broken arbitrarily. Create two new vertices $u', u''$ and add the edge $(u', u'')$ of weight $\omega_G(u, v_1)$. Then add the edge $(u, u')$ of weight zero, and replace each edge $(u, v_i)$ by the edge $(u'', v_i)$ of weight $\omega_G(u, v_i)-\omega_G(u, v_1) \geq 0$. It is straightforward to verify all of the required conditions.
	\end{proof}

	\subsection{Reviewing the Shortcutting Algorithm by \cite{li2025shortcutting}}
	
	We begin by reviewing the shortcutting approach by \cite{li2025shortcutting}. For simplicity in this section, we will assume that every vertex has in- and out-degree $O(m/n)$ and show that one iteration of shortcutting takes $\tilde{O}(m\sqrt{k})$ time. \cite{li2025shortcutting} used the trivial bound $m\le n^2$, giving their result.

	The first step is to compute a reweighting $\phi$ so that the negative paths in the graph have more structure. A similar betweenness reduction procedure has been key to all previous work beating Bellman-Ford. We remark that the notation below is slightly different from \cite{li2025shortcutting}.
	
	\begin{definition}
		For $s,t\in V$, the (weak) betweenness of $(s,t)$, denoted $\text{BW}(s,t)$ is the number of negative vertices $r\in N$ such that $d^0(s,r)+d^{1}(r,t)<0$.
	\end{definition}
	\begin{lemma}[Betweenness Reduction, informal]\label[lemma]{lem:betweenness-reduction}
		Consider a graph $G$ with $k$ negative vertices. For any parameter $b\ge1$, we can compute valid potentials $\phi$ such that with high probability, all pairs $(s,t)\in V\times V$ have betweenness at most $k/b$ under the new weights $w_\phi$. The algorithm makes one call to negative-weight single-source shortest paths on a graph with $m$ edges, $n$ vertices, and $O(b\log n)$ negative vertices, and takes $O(m\log{n})$ additional time.
	\end{lemma}
	
	In the resulting graph with low betweenness, they then compute forward and backward Dijkstra searches from each negative edge. They show that these searches can be implemented simultaneously for all negative vertices efficiently in a graph with low betweenness. 
	
	\begin{lemma}[Informal version of Lemma 8 in \cite{li2025shortcutting}]\label[lemma]{lem:dijkstra-both-ways}
		There is an algorithm that, for any given negative vertex $r\in N$, computes a number $\Delta_r$ and two sets $V^{\textup{out}}_r$ and $V^{\textup{in}}_r$ such that
		\begin{enumerate}
			\item $d^1(r,v)<-\Delta_r$ for all $v\in V^{\textup{out}}_r$,
			\item $d^1(r,v)>-\Delta_r$ for all $v\not\in V^{\textup{out}}_r$,
			\item $d^0(v,r)<\Delta_r$ for all $v\in V^{\textup{in}}_r$,
			\item $d^0(v,r)>\Delta_r$ for all $v\not\in V^{\textup{in}}_r$,
			\item $\big||V^{\textup{out}}_r|-|V^{\textup{in}}_r|\big|\le 1.$
		\end{enumerate}
		The algorithm runs in time $O(\frac{m}{n}(|V^{\textup{out}}_r|+|V^{\textup{in}}_r|)\log{n})$. Moreover, the algorithm can output the values of $d^1(r,v)$ for all $v\in V^{\textup{out}}_r$ and $d^0(v,r)$ for all $v\in V^{\textup{in}}_r$.
	\end{lemma}
	
	\begin{lemma}[Lemma 9 of \cite{li2025shortcutting}]\label[lemma]{lem:total-size-bound}
		Under the betweenness reduction guarantee of \Cref{lem:betweenness-reduction}, we have $\sum_{r\in N}(|V_r^{\textup{out}}|+|V_r^{\textup{in}}|)^2\le O(kn^2/b)$ and $\sum_{r\in N}(|V_r^{\textup{out}}|+|V_r^{\textup{in}}|)\le O(kn/\sqrt b)$.
	\end{lemma}
	
	Choosing $b=k/\textup{poly}\log{n}$, \Cref{lem:dijkstra-both-ways,lem:total-size-bound} imply that we can compute sets $V_r^{\textup{out}}$ and $V_r^{\textup{in}}$ for all $r\in N$ in time $\tilde{O}(m\sqrt{k})$. Using these sets, the algorithm \emph{shortcuts} the graph by adding \emph{Steiner vertices} with the following edges to the graph; see Figure~\ref{fig:shortcut} for a visual reference.
	\begin{enumerate}
		\item For each negative vertex $r\in N$, create a new vertex $\tilde r$.
		\item For each $r\in N$, $v\in V_r^{\textup{out}}\cup\{r\}$, and $w\in N^{\textup{out}}(v)$, add the edge $(\tilde r,w)$ of weight $d^1(r,v)+\Delta_r+w(v,w)$ \emph{only if it is non-negative},\label{item:shortcut-2}
		\item For each $r\in N$, $v\in V_r^{\textup{in}}\cup\{r\}$, and $u\in N^{\textup{in}}(v)$, add the edge $(u,\tilde r)$ of weight $w(u,v)+d^0(v,r)-\Delta_r$ \emph{only if it is non-negative},\label{item:shortcut-3}
		\item For each negative edge $(r,r')$ and each $u\in V_r^{\textup{out}}\cup\{r\}$, if there is a (unique) negative out-edge $(u,v)$ of $u$, add the edge $(r,v)$ of weight $d^1(r,u)+w(u,v)$, and\label{item:shortcut-4}
		\item For each negative edge $(r,r')$ and each $v\in V_r^{\textup{in}}\cup\{r\}$, if there is a (unique) negative in-edge $(u,v)$ of $v$, add the edge $(u,r')$ of weight $w(u,v)+d^0(v,r)+w(r,r')$.\label{item:shortcut-5}
	\end{enumerate}
	
	\begin{figure}
		\centering
		\includegraphics[scale=1]{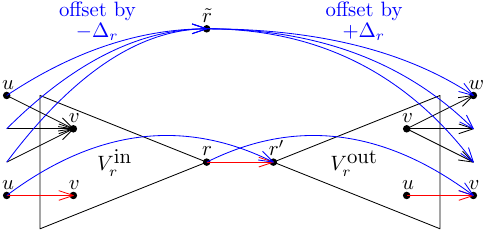}
		\caption{The shortcut construction from \cite{li2025shortcutting}. Existing negative edges are colored red, and the new shortcut edges are colored blue.}
		\label{fig:shortcut}
	\end{figure}
	
	The authors show that adding these shortcut edges suffices to reduce the number of hops in all shortest paths by a constant factor. Iterating this shortcutting procedure gives the full algorithm.
	
	
	\begin{lemma}\label{lem:shortcut}
		Consider any $s,t\in V$ and a shortest $(s,t)$-path $P$ with $h$ negative edges. After the shortcutting procedure above, there is an $(s,t)$-path with at most $h-\lfloor h/3\rfloor$ negative edges, and whose weight is at most that of $P$.
	\end{lemma}
	
	\subsection{Technical Overview}\label{sec:technical-overview}
	
	\paragraph{Improved Betweenness Reduction.} The primary bottleneck of the current approach is that there may be $\Omega(m\sqrt{k})$ edges added in Steps 2 and 3 of shortcutting. This is due to the fact that existing betweenness reduction analysis can only show an $O(n\sqrt{k})$ bound on the total sizes of the searches, and each of those vertices may have $\Theta(m/n)$ neighbors. Our first technical contribution is an improved bound on the search sizes, where we show the total search size is $n^{1+o(1)}$ under the \emph{same} betweenness reduction algorithm as before. Plugging this bound into the sketch above gives a shortcutting algorithm which adds only $m^{1+o(1)}$ edges. 
	
	The prior work bound the search size as follows: they first bound the total squared search sizes based on the total betweenness over all pairs $s,t\in V$: $\sum_{r\in N}|V_r|^2\le \sum_{s,t\in V}\text{BW}(s,t)$. Under the betweenness reduction guarantee, the betweenness of each pair is at most $n^{o(1)}$, so the total betweenness is at most $n^{2+o(1)}$. Finally, they apply Cauchy-Schwarz to bound the total search size by $\sum_{r\in N}|V_r|\le n^{1.5+o(1)}$. Informally, we go beyond this $n^{1.5}$ barrier by bounding the total search size directly: we use a \emph{subset betweenness reduction} argument to bound $\sum_{r\in N}|V_r|\le n^{1+o(1)}$. A self-contained proof of the improved betweenness reduction analysis is given in \Cref{sec:improved-betweenness-reduction}.
	
	\paragraph{Unfolding Lemma.} Applying this improved bound naively still (more than) doubles the number of edges on each shortcutting step. Since we need to shortcut $\Omega(\log{n})$ iterations, this still blows up the number of edges added to the graph. Our second contribution is a structural result which helps argue that this blowup does not happen. We show that under betweenness reduction, the number of edges added can be bounded by $m_0^{1+o(1)}$, where $m_0$ is the number of edges in the original graph before any shortcutting is done. 
	
	The idea is to charge each edge $(u,v)$ on the boundary of the searches to some original edge $(\tilde{u},\tilde{v})$ on the boundary of the searches. The number of such original edges is at most $m_0^{1+o(1)}$, so this suffices as long as we don't charge any original edge too many times. Any original edge on the boundary can simply be charged to itself, so consider a slightly harder case where $u\in V_r^{\textup{out}}$ is a Steiner vertex and $v$ is an original vertex. Suppose that there have already been $\tau$ iterations of shortcutting and $u$ was added in this final iteration. 
	
	Since $(u,v)$ is an edge incident to Steiner vertex $u$, we know that it was added to shortcut some $1$-hop path which existed after $\tau-1$ iterations of shortcutting. We can \emph{unfold} the edge into this $1$-hop path which it was designed to shortcut. The final edge in this unfolded path can now be iteratively unfolded until it becomes an edge $(\tilde u,\tilde v)$ between original vertices. This unfolded path has $O(\tau)$ negative hops, which is a certificate that $\tilde{u}$ is in the $\operatorname{polylog}(n)$-hop version of the forward search. The betweenness reduction algorithm guarantees that the $\operatorname{polylog}(n)$-hop searches are also bounded in size, so we can bound the number of such $(\tilde u,\tilde v)$ original edges.

	\paragraph{Other Technical Details.} The unfolding lemma is presented in \Cref{sec:unfolding}, but setting up the proof brings several additional technical details. The most prominent technicality is formalizing this intuition that we can unfold the paths. We formalize this as an invariant of our Steiner vertices in \Cref{sec:steiner}, and show that the invariant is satisfied by our definition of Steiner shortcut edges. In \Cref{sec:recursive-betweenness,sec:maintain-invariant}, we show that the invariants are preserved in our recursive calls for betweenness reduction and when we apply reweightings to the graph. To formalize the charging argument, we also need a weighted version of the forward and backward searches, which is presented in \Cref{sec:forward-backward-search}.

	\section{Improved Analysis of Betweenness Reduction}\label{sec:improved-betweenness-reduction}

	In this section, we prove our main guarantee of the betweenness reduction. Fix parameter $h$ throughout the section; we say that a potential function $\phi$ \emph{fixes} a subset of negative vertices $U\subseteq N$ if for all $s,t\in V$ and $r\in U$, we have 
    $d^h_\phi(s,r)+d^h_\phi(r,t)\ge0$. 
	\begin{theorem}\label{thm:betweenness}
		Let $U\subseteq N$ be a random sample where each vertex in $N$ is sampled with probability $1/b$. With high probability over the randomness of $U$, the following holds. Consider any potential function $\phi$ that fixes $U$, and for each negative vertex $r\in N$, consider any $V_r^{\textup{in}},V_r^{\textup{out}}\subseteq V$ such that for all $s\in V_r^{\textup{in}},\,t\in V_r^{\textup{out}}$, we have $d^h_\phi(s,r)+d^h_\phi(r,t)<0$. Then,
		\[ \sum_{r\in N}\min\{|V_r^{\textup{in}}|,|V_r^{\textup{out}}|\}\le\tilde{O}(nb) .\]
	\end{theorem}
	
	In this section, we work with subsets of $V\times N\times V$. For a given subset $\mathcal X\subseteq V\times N\times V$, we define ``slices'' of $\mathcal X$ as follows: for a tuple $(s,r,t)\in(V\cup\{*\})\times(N\cup\{*\})\times(V\cup\{*\})$ of vertices with possible ``wildcards'' $*$, define $\mathcal X[s,r,t]\subseteq\mathcal X$ as the following ``slice'' of $\mathcal X$:
	\[ \mathcal X[s,r,t]=\{(s',r',t')\in\mathcal X:(s'=s\textup{ or } s=*)\textup{ and }(r'=r\textup{ or } r=*)\textup{ and }(t'=t\textup{ or } t=*)\} .\]
	For example, for a vertex $r\in N$, $\mathcal X[*,r,*]$ is the set of tuples $\{(s',r',t')\in\mathcal X:r'=r\}$. We also extend this notation to subsets $S,T\subseteq V$ and $r\in N\cup\{*\}$:
	\[ \mathcal X[S,r,T]=\{(s',r',t')\in\mathcal X:s'\in S\textup{ and }(r'=r\textup{ or }r=*)\textup{ and }t'\in T\} .\]
	
	For a potential function $\phi$, let $\mathcal X_\phi\subseteq V\times N\times V$ be the set of tuples $(s,r,t)\in V\times N\times V$ for which $d_\phi^h(s,r)+d^h_\phi(r,t)<0$.
	
	We now state our main technical lemma.
	
	\begin{restatable}{lemma}{BetweennessSingleStep}\label{lem:betweenness-single-step}
		Let $U\subseteq N$ be a random sample where each vertex in $N$ is sampled with probability $1/b$. With high probability over the randomness of $U$, the following holds. Consider any parameter $\ell\ge1$ and any subset $\mathcal Y\subseteq V\times N\times V$ such that
		\begin{enumerate}
			\item $|\mathcal Y[*,r,*]|\le n^2/\ell^2$ for all $r\in N$,\label{item:sparsity-1}
			\item $|\mathcal Y[s,r,*]|\le n/\ell$ for all $s\in V$ and $r\in N$, and\label{item:sparsity-2}
			\item $|\mathcal Y[*,r,t]|\le n/\ell$ for all $r\in N$ and $t\in V$.\label{item:sparsity-3}
		\end{enumerate}
		Then, for any potential function $\phi$ that fixes $U$, we have $|\mathcal X_\phi\cap\mathcal Y|\le O(\frac{n^2}\ell\cdot\frac{b\log^3n}{\log\log n})$.
	\end{restatable}
	
	We first show how to ``bootstrap'' \Cref{lem:betweenness-single-step} to prove \Cref{thm:betweenness}.
	\begin{proof}[Proof (\Cref{lem:betweenness-single-step}$\implies$\Cref{thm:betweenness})]
		We assume that \Cref{lem:betweenness-single-step} holds for the random sample $U$, which occurs with high probability. The rest of the proof succeeds with probability $1$.
		
		Let $\phi$ be any potential function that fixes $U$. Initialize $\mathcal Y_0$ as the set of tuples $(s,r,t)\in V\times N\times V$ such that $r\in N,\,s\in V_r^{\textup{in}},\,t\in V_r^{\textup{out}}$, and set $\ell_i\gets2^i$ for all $i\in\{0,1,2,\ldots,\lceil\log_2n\rceil\}$. By the guarantee of $V_r^{\textup{in}},V_r^{\textup{out}}$ from \Cref{thm:betweenness}, we have $\mathcal Y_0\subseteq\mathcal X_\phi$. For iterations $i\in\{1,2,3,\ldots,\lceil\log_2n\rceil\}$, perform the following:
		\begin{enumerate}
			\item Apply \Cref{lem:betweenness-single-step} on $\mathcal Y_{i-1}$ and $\ell_{i-1}$, which must satisfy the requirements of \Cref{lem:betweenness-single-step}. We also ensure that $\mathcal Y_{i-1}\subseteq\mathcal X_\phi$. Since $\phi$ fixes $U$, we have $|\mathcal Y_{i-1}|=|\mathcal X_\phi\cap\mathcal Y_{i-1}|\le O(\frac{n^2}{\ell_{i-1}}\cdot\frac{b\log^3n}{\log\log n})$.
			\item Initialize $\mathcal Y_i\gets\mathcal Y_{i-1}$.
			\item Iteratively perform the following until no further action is possible:
			\begin{enumerate}
				\item If $|\mathcal Y_i[*,r,*]|\ge n^2/\ell_i^2$ for some $r\in N$, then update $\mathcal Y_i\gets\mathcal Y_i\setminus\mathcal Y_i[*,r,*]$.\label{item:remove-r-1}
				\item If $|\mathcal Y_i[s,r,*]|\ge n/\ell_i$ for some $s\in V,\,r\in N$, then update $\mathcal Y_i\gets\mathcal Y_i\setminus\mathcal Y_i[s,r,*]$.\label{item:remove-r-2}
				\item If $|\mathcal Y_i[*,r,t]|\ge n/\ell_i$ for some $r\in N,\,t\in V$, then update $\mathcal Y_i\gets\mathcal Y_i\setminus\mathcal Y_i[*,r,t]$.\label{item:remove-r-3}
			\end{enumerate}
			This guarantees that the next call to \Cref{lem:betweenness-single-step} satisfies the necessary requirements, as well as $\mathcal Y_i\subseteq\mathcal Y_{i-1}\subseteq\mathcal X_\phi$. Furthermore, for the last iteration $i=\lceil\log_2n\rceil$, we have $\mathcal Y_i=\emptyset$.
		\end{enumerate}
		We account for the updates in steps~(\ref{item:remove-r-1}) to (\ref{item:remove-r-3}) as follows:
		\begin{enumerate}[label=(\alph*)]
			\item There are at most $\frac{|\mathcal Y_{i-1}|}{n^2/\ell_i^2}=O(\ell_i\cdot\frac{b\log^3n}{\log\log n})$ negative vertices $r\in N$ processed in step~(\ref{item:remove-r-1}). For each such $r\in N$, since $|\mathcal Y_{i-1}[*,r,*]|\le n^2/\ell_{i-1}^2$, we removed at most $n^2/\ell_{i-1}^2$ entries from $\mathcal Y_i$. We charge to $r$ the square-root of the number of entries removed, which is at most $n/\ell_{i-1}$. The total charge is at most $O(n\cdot\frac{b\log^3n}{\log\log n})$.\label{item:charge-r-1}
			\item There are at most $\frac{|\mathcal Y_{i-1}|}{n/\ell_i}=O(n\cdot\frac{b\log^3n}{\log\log n})$ tuples $(s,r)\in V\times N$ processed in step~(\ref{item:remove-r-2}). For each such tuple $(s,r)$, we charge $1$ to $r$.\label{item:charge-r-2}
			\item There are at most $\frac{|\mathcal Y_{i-1}|}{n/\ell_i}=O(n\cdot\frac{b\log^3n}{\log\log n})$ tuples $(r,t)\in N\times V$ processed in step~(\ref{item:remove-r-3}). For each such tuple $(r,t)$, we charge $1$ to $r$.\label{item:charge-r-3}
		\end{enumerate}
		The total charge over all $\lceil\log_2n\rceil$ iterations is $O(n\cdot\frac{b\log^4n}{\log\log n})$. At the same time, we lower bound the charge as follows.
		
		\begin{subclaim}\label{clm:charge-r}
			Each negative vertex $r\in N$ is charged at least $\min\{|V_r^{\textup{in}}|,|V_r^{\textup{out}}|\}/2$.
		\end{subclaim}
		\begin{subproof}
			Fix a negative vertex $r\in N$, and consider the set $\mathcal Y_i[*,r,*]$ throughout the process, which only loses elements over time. Initially, $\mathcal Y_0[*,r,*]=V_r^{\textup{in}}\times\{r\}\times V_r^{\textup{out}}$, which we view as a matrix with rows indexed by $V_r^{\textup{in}}$ and columns indexed by $V_r^{\textup{out}}$. If step~(\ref{item:remove-r-2}) is executed on tuple $(s,r)$, then the row $\mathcal Y_i[s,r,*]$ is removed, and if step~(\ref{item:remove-r-3}) is executed on tuple $(r,t)$, then the column $\mathcal Y_i[*,r,t]$ is removed. In particular, each tuple $(s,r)$ and $(r,t)$ is responsible for charging $r$ at most once from~\ref{item:charge-r-2} and \ref{item:charge-r-3}. Finally, if step~(\ref{item:remove-r-1}) is executed on $r$, then the entire matrix is cleared.
			
			If there are at least $\min\{|V_r^{\textup{in}}|,|V_r^{\textup{out}}|\}/2$ charges to $r$ from~\ref{item:charge-r-2} or \ref{item:charge-r-3}, then we are done. Otherwise, the remaining elements in $\mathcal Y_i[*,r,*]$ must be removed by a charge from~\ref{item:charge-r-1}. At this point, the remaining matrix has dimensions at least $|V_r^{\textup{in}}|-\min\{|V_r^{\textup{in}}|,|V_r^{\textup{out}}|\}/2$ by $|V_r^{\textup{out}}|-\min\{|V_r^{\textup{in}}|,|V_r^{\textup{out}}|\}/2$, so there are at least
			\[ (|V_r^{\textup{in}}|-\min\{|V_r^{\textup{in}}|,|V_r^{\textup{out}}|\}/2)\cdot(|V_r^{\textup{out}}|-\min\{|V_r^{\textup{in}}|,|V_r^{\textup{out}}|\}/2)\ge(\min\{|V_r^{\textup{in}}|,|V_r^{\textup{out}}|\}/2)^2 \]
			many elements remaining in $\mathcal Y_i[*,r,*]$. It follows that $r$ is charged at least $\min\{|V_r^{\textup{in}}|,|V_r^{\textup{out}}|\}/2$, as promised.
		\end{subproof}
		By \Cref{clm:charge-r}, the total charge is at least $\sum_{r\in N}\min\{|V_r^{\textup{in}}|,|V_r^{\textup{out}}|\}/2$. At the same time, the total charge is at most $O(n\cdot\frac{b\log^4n}{\log\log n})$, so $\sum_{r\in N}\min\{|V_r^{\textup{in}}|,|V_r^{\textup{out}}|\}/2\le O(n\cdot\frac{b\log^4n}{\log\log n})$, completing the proof of \Cref{thm:betweenness}.
	\end{proof}
	
	It remains to prove \Cref{lem:betweenness-single-step}. We begin with the following \emph{subset} betweenness reduction guarantee, which is a generalization of the standard betweenness reduction to subsets $S,T\subseteq V$. (Observe that the statement is the same for $\ell=1$.)
	\begin{lemma}\label{clm:distinct-bound}
		With high probability over the randomness of $U$, the following holds for all integers $\ell\in[n]$ and all subsets $S,T\subseteq V$ with $|S|=|T|=\ell$: for any potential function $\phi$ that fixes $U$, the set $R_\phi=\{r:\mathcal X_\phi[S,r,T]\ne\emptyset\}$ has size at most $O(b\ell\log n)$.
	\end{lemma}
	\begin{proof}
		We want to perform a union bound over all $S,T\subseteq V$ and all possible potential functions $\phi$. Since the number of potential functions is infinite, we instead union bound over all possible \emph{distinct} sets $\mathcal X_\phi[S,*,T]$.
		
		\begin{subclaim}\label{clm:signing}
			For any $S,T\subseteq V$, the number of possible distinct sets $\mathcal X_\phi[S,*,T]$ over all possible potential functions $\phi$ is $O((|S|\cdot|T|\cdot|N|)^{|S\cup T|})$.
		\end{subclaim}
		\begin{subproof}
			For a given potential function $\phi$, we have $(s,r,t)\in\mathcal X_\phi[S,*,T]$ if and only if $d^h_\phi(s,r)+d^h_\phi(r,t)<0$, or equivalently $d^h(s,r)+d^h(r,t)+\phi(s)-\phi(t)<0$. Treating $\phi(s)$ and $\phi(t)$ as variables and $d^h(s,r)+d^h(r,t)$ as a constant, we conclude that $(s,r,t)\in\mathcal X_\phi[S,*,T]$ if and only if a certain linear inequality on variables $\phi(s)$ and $\phi(t)$ is satisfied. Each linear inequality corresponds to a hyperplane in $\mathbb R^{S\cup T}$, and there are $|S|\cdot|T|\cdot|N|$ many linear inequalities over all $s\in S,\,t\in T,\,r\in N$. It is well known that an arrangement of $x$ hyperplanes in $\mathbb{R}^d$ has $O(x^d)$ cells. It follows that the number of cells is $O((|S|\cdot|T|\cdot|N|)^{|S\cup T|})$, and each region corresponds to a possible set $\mathcal X_\phi[S,*,T]$.
		\end{subproof}
		Using \Cref{clm:signing}, we now prove \Cref{clm:distinct-bound}. Consider subsets $S,T\subseteq V$ with $|S|=|T|=\ell$, and consider a set $\mathcal X\subseteq V\times N\times V$ for which the set $R=\{r:\mathcal X[S,r,T]\ne\emptyset\}$ has size at least $Cb\ell\ln n$, where $C>0$ is a large enough constant. Our goal is to bound the probability, over the random choice of $U$, that there \emph{exists} a potential function $\phi$ that fixes $U$ and satisfies $\mathcal X_\phi[S,*,T]=\mathcal X[S,*,T]$. Then, we finish by taking a union bound over all possible $S,T\subseteq V$ and $\mathcal X[S,*,T]$.
		
		The key claim is that if such a potential function exists, then $U$ must be disjoint from $R$. Otherwise, if $r\in U\cap R$, then consider $s\in S$ and $t\in T$ with $(s,r,t)\in\mathcal X$. We have $d^h_\phi(s,r)+d^h_\phi(r,t)<0$ since $(s,r,t)\in\mathcal X_\phi$, but also $d^h_\phi(s,r)+d^h_\phi(r,t)\ge0$ since $\phi$ fixes $U$, a contradiction.
		
		Therefore, if such a potential function $\phi$ exists, then the sampled subset $U$ must be disjoint from $R$. Each vertex in $R$ is sampled into $U$ with probability $1/b$, so the probability that none are sampled is
		\[ \left(1-\frac1b\right)^{|R|}\le\left(1-\frac1b\right)^{Cb\ell\ln n}\le\exp\left(-\frac1b\cdot Cb\ell\ln n\right)=n^{-C\ell} .\]
		Finally, taking a union bound over all $S,T\subseteq V$ with $|S|=|T|=\ell$ and possible set $\mathcal X[S,*,T]$, using the bound from \Cref{clm:signing}, the overall probability of failure is at most
		\[ n^\ell\cdot n^\ell\cdot O((n^3)^{2\ell})\cdot n^{-C\ell} ,\]
		which is polynomially small when $C>0$ is large enough. Taking a union bound over $\ell\in[n]$ finishes the proof.
	\end{proof}
	
	We now prove \Cref{lem:betweenness-single-step}, restated below.
	\BetweennessSingleStep*
	\begin{proof}
		We assume that \Cref{clm:distinct-bound} holds for the random sample $U$, which occurs with high probability. The rest of the proof succeeds with probability $1$.
		
		Consider any parameter $\ell\ge1$ and any subset $\mathcal Y\subseteq V\times N\times V$ satisfying assumptions~(\ref{item:sparsity-1}) to~(\ref{item:sparsity-3}). To bound $|\mathcal X_\phi\cap\mathcal Y|$, it will be more convenient to bound $|(\mathcal X_\phi\cap\mathcal Y)[S_i,*,T_j]|$ where $V\times V$ is partitioned into rectangular subsets $S_i\times T_j$ with desirable properties:
		\begin{subclaim}\label{clm:r-bound}
			There exist partitions $\{S_1,S_2,\ldots,S_{n/\ell}\}$ and $\{T_1,T_2,\ldots,T_{n/\ell}\}$ of $V$, each with $n/\ell$ subsets of size $\ell$, such that $|\mathcal Y[S_i,r,T_j]|\le O(\frac{\log^2n}{\log\log n})$ for each $r\in N$ and $i,j\in[n/\ell]$.
		\end{subclaim}
		\begin{subproof}
			Let $\{S_1,S_2,\ldots,S_{n/\ell}\}$ and $\{T_1,T_2,\ldots,T_{n/\ell}\}$ be two independent random partitions of $V$, each with $n/\ell$ subsets of size $\ell$. It suffices to show that the statement holds with nonzero probability.
			
			For the entire proof, fix $r\in N$ and $i,j\in[n/\ell]$. For each $s\in S_i$, we have $|\mathcal Y[s,r,*]|\le n/\ell$ by assumption~(\ref{item:sparsity-2}) of \Cref{lem:betweenness-single-step}. Each tuple $(s,r,t)\in\mathcal Y[s,r,*]$ satisfies $t\in T_j$ with probability $\ell/n$, so by a Chernoff bound (with negative correlation), we have $|\mathcal Y[s,r,T_j]|\le O(\log n)$ with high probability. For each $t\in T_j$, by the same argument with assumption~(\ref{item:sparsity-3}), we have $|\mathcal Y[S_i,r,t]|\le O(\log n)$ with high probability.
			
			Next, we view $\mathcal Y[*,r,*]$ as a bipartite graph with (a copy of) the vertex set $V$ on each side. Formally, we define the graph $H_{\mathcal Y,r}$ on vertex set $\{v_1:v\in V\}\sqcup\{v_2:v\in V\}$ with the edge set $\{(s_1,t_2):(s,r,t)\in\mathcal Y]$. Consider the induced subgraph $H_{\mathcal Y,r}[W]$ with $W=\{v_1:v\in S_i\}\sqcup\{v_2:v\in T_j\}$. By the previous argument, the maximum degree of $H_{\mathcal Y,r}[W]$ is $O(\log n)$ with high probability. Observe that $|\mathcal Y[S_i,r,T_j]|$ is exactly the number of edges of $H_{\mathcal Y,r}[W]$, and a greedy matching procedure gives a matching of size $\Omega(|\mathcal Y[S_i,r,T_j]|/\log n)$. Therefore, in order to upper bound $|\mathcal Y[S_i,r,T_j]|$, it suffices to upper bound the maximum matching of $H_{\mathcal Y,r}[W]$.
			
			Let $q\ge1$ be an integer. By assumption~(\ref{item:sparsity-1}), the initial bipartite graph $H_{\mathcal Y,r}$ has at most $n^2/\ell^2$ edges, so the number of matchings of size $q$ is at most $\binom{n^2/\ell^2}q\le(\frac{e\cdot n^2/\ell^2}q)^q=(\frac eq)^q\cdot(\frac n\ell)^{2q}$. Each such matching is included in $H_{\mathcal Y,r}[W]$ if and only if each edge $(s_1,t_2)$ in the matching satisfies $s\in S_i$ and $t\in T_j$. For each edge $(s_1,t_2)$ in the matching, the probability that $s\in S_i$ and $t\in T_j$ is exactly $(\ell/n)^2$. Since each edge in the matching has distinct endpoints, the probabilities are negatively correlated over the edges, so the overall probability that a matching is included in $H_{\mathcal Y,r}[W]$ is at most $(\ell/n)^{2q}$. Taking a union bound over all matchings, the probability that $H_{\mathcal Y,r}[W]$ has a matching of size $q$ is at most $(\frac eq)^q\cdot(\frac n\ell)^{2q}\cdot(\frac\ell n)^{2q}=(\frac eq)^q$. For $q=\Theta(\frac{\log n}{\log\log n})$, the subgraph $H_{\mathcal Y,r}[W]$ has maximum matching at most $q-1$ with high probability. We conclude that $|\mathcal Y[S_i,r,T_j]|\le O((q-1)\log n)=O(\frac{\log^2n}{\log\log n})$ with high probability. Taking a union bound over all $r\in N$ and $i,j\in[n/\ell]$ proves the subclaim.
		\end{subproof}
		
		For the rest of the proof, consider the partitions $\{S_1,S_2,\ldots,S_{n/\ell}\}$ and $\{T_1,T_2,\ldots,T_{n/\ell}\}$ guaranteed by \Cref{clm:r-bound}. 
		
		\begin{subclaim}\label{clm:submatrix-bound}
			For any potential function $\phi$ that fixes $U$, we have $|(\mathcal X_\phi\cap\mathcal Y)[S_i,*,T_j]|\le O(\frac{b\ell\log^3n}{\log\log n})$ for each $i,j\in[n/\ell]$.
		\end{subclaim}
		\begin{subproof}
			By \Cref{clm:distinct-bound}, the set $R_\phi=\{r:\mathcal X_\phi[S_i,r,T_j]\ne\emptyset\}$ has size at most $O(b\ell\log n)$. By \Cref{clm:r-bound}, we have $|\mathcal Y[S_i,r,T_j]|\le O(\frac{\log^2n}{\log\log n})$ for each $r\in N$. Both of these statements still hold with $\mathcal X_\phi$ and $\mathcal Y$ replaced by $\mathcal X_\phi\cap\mathcal Y$, so we conclude that $|(\mathcal X_\phi\cap\mathcal Y)[S_i,*,T_j]|\le O(\frac{b\ell\log^3n}{\log\log n})$.
		\end{subproof}
		
		Finally, we finish the proof of \Cref{lem:betweenness-single-step}. The set $\mathcal X_\phi\cap\mathcal Y$ is partitioned into $(n/\ell)^2$ many subsets $(\mathcal X_\phi\cap\mathcal Y)[S_i,*,T_j]$, each of size $O(\frac{b\ell\log^3n}{\log\log n})$ by \Cref{clm:submatrix-bound}, for a total size of $O(\frac{n^2}\ell\cdot\frac{b\log^3n}{\log\log n})$.
	\end{proof}
	
	\subsection{Weighted Betweenness Reduction}
	
	We now introduce a weighted version of our betweenness reduction guarantee in \Cref{thm:betweenness}. Assume that every vertex $v\in V$ is assigned a non-negative weight $\lambda(v)$. We show an analogue of \Cref{thm:betweenness} where the sizes $|V_r^{\textup{in}}|$, $|V_r^{\textup{out}}|$ are replaced by the total weight of vertices in $V_r^{\textup{in}}$, $V_r^{\textup{out}}$. 
	
	\begin{corollary} \label{lem:betweenness-weighted}
		Consider a non-negative weight $\lambda(v)$ on every vertex $v$ in the graph. Let $\Lambda=\sum_{v\in V}\lambda(v)$ denote the total weight. Let $U\subseteq N$ be a random sample where each vertex in $N$ is sampled with probability $1/b$. With high probability over the randomness of $U$, the following holds. Consider any potential function $\phi$ that fixes $U$, and for each negative vertex $r\in N$, consider any $V_r^{\textup{in}},V_r^{\textup{out}}\subseteq V$ such that for all $s\in V_r^{\textup{in}},\,t\in V_r^{\textup{out}}$, we have $d^h_\phi(s,r)+d^h_\phi(r,t)<0$. Then,
		\[ \sum_{r\in N}\min\Big\{\sum_{v\in V_r^{\textup{in}}}\lambda(v),\sum_{v\in V_r^{\textup{out}}}\lambda(v)\Big\}\le\tilde{O}(\Lambda b) .\]
	\end{corollary}
	
	\begin{proof}
		We reduce to the unweighted case in \Cref{thm:betweenness}. Define the rounded weights\footnote{If $\Lambda = 0$ then the corollary holds immediately, so we can assume $\Lambda > 0$.} $\lambda'(v)=\lfloor \frac{\lambda(v)}\Lambda\cdot n\rfloor+1$. It suffices to prove that
		\[ \sum_{r\in N}\min\Big\{\sum_{v\in V_r^{\textup{in}}}\lambda'(v),\sum_{v\in V_r^{\textup{out}}}\lambda'(v)\Big\}\le\tilde{O}(nb) .\]
		Next, construct an auxiliary graph $G'$ as follows. Starting with the original graph $G$, for each vertex $v$, add $\lambda'(v)-1$ extra copies of $v$, and attach them to the original vertex $v$ with bidirectional zero-weight edges. This auxiliary graph $G'$ has $O(n)$ vertices since $\sum_u \lambda'(u)=O(n)$.
		
		Consider a potential function $\phi$ that fixes $U$ in $G$. We extend $\phi$ to a potential function that fixes $U$ in $G'$ by setting $\phi(v') = \phi(v)$ for all vertex copies $v'$ of each vertex $v$. Since each $v$ has $\lambda'(v)$ copies (including the original one), \Cref{thm:betweenness} implies that
		\[ \sum_{r\in N}\min\Big\{\sum_{v\in V_r^{\textup{in}}}\lambda'(v),\sum_{v\in V_r^{\textup{out}}}\lambda'(v)\Big\}\le\tilde{O}(nb) ,\]
		which completes the proof.
		%
		%
	\end{proof}
	
	\section{Preliminaries}

	Let $G=(V,E)$ be a directed graph with a set $N$ of negative vertices; our goal is to compute single-source shortest paths in $G$. For technical reasons in the recursive subproblems, we are also given as input a set of auxiliary edges $E_{\textup{fake}}$. We will need to include these edges in our betweenness reduction and forward/backward searches in order for our \emph{unfolding lemma} to be valid and to ensure that the shortest path distance returned by our recursive problems satisfies additional properties that we need. We emphasize that each of our recursive problems solves the single-source shortest path problem for the graph without any of the auxiliary edges $E_{\textup{fake}}$. In the top level of recursion, we simply set $E_{\textup{fake}}=\emptyset$ to solve the original problem.

	For these auxiliary edges, we define $\tilde d^h(u,v)$ to be the $h$-hop shortest path distance that may use auxiliary edges as hops. For shortest path distances in the graph without auxiliary edges, we still use $d(u,v)$ and $d^h(u,v)$. For a subset $S\subseteq V$, we also let $d(S,v)$ and $d^h(S,v)$ be the shortest path and $h$-hop shortest paths from any vertex in $S$ to $v$. Finally, when we say the volume of a subset $S$ of vertices, we are referring to the sum of the degrees of those vertices.
	
	For simplicity, we maintain the following invariant on (both real and auxiliary) negative edges:
	
	\begin{enumerate}
		\item[{\crtcrossreflabel{(I1)}[item:negative-edge-invariant-1]}] Every vertex is adjacent to at most one (real or auxiliary) negative edge. If a vertex $r$ is the tail of a real negative edge, then that edge is the only outgoing edge from $r$. In addition, if a vertex $r'$ is the head of a real negative edge, then that edge is the only incoming edge to $r'$.
	\end{enumerate}
	
	For the original graph, we use \Cref{lem:one-negative-outgoing-edge} to ensure this property. During the algorithm, we maintain the invariant whenever we construct new auxiliary graphs.

	\subsection{Steiner Vertices}\label{sec:steiner}
	
	Our shortcutting algorithm will require adding additional nodes, called Steiner vertices, to the graph in order to shortcut more efficiently. In our algorithms, the Steiner vertices are the $\tilde{r}$ vertices from \cite{li2025shortcutting} added in the shortcutting step. These vertices will then be passed into the recursive SSSP instances in our algorithm. Below, we enforce some important invariants of Steiner vertices.
	
	As input, the vertex set $V$ will be partitioned into subsets $V=S_0\sqcup S_1\sqcup\cdots\sqcup S_{\tau}$, where $\tau$ is the number of shortcutting iterations already performed. Informally, $S_0$ are the original vertices and vertex copies and $S_i$ for $i\ge1$ are the Steiner vertices added on the $i^{th}$ iteration of shortcutting. To avoid clutter, we use the notation $S_{<i}=S_0\sqcup\cdots\sqcup S_{i-1}$. We now formally state the invariants we wish to maintain on the Steiner vertices:
	
	\begin{enumerate}
		\item[{\crtcrossreflabel{(I2)}[item:steiner-invariant-1]}] There are no edges between Steiner vertices in the same level. In other words, there are no edges $(u,v)$ for $u,v\in S_i$, $i\neq 0$.
		\item[{\crtcrossreflabel{(I3)}[item:steiner-invariant-2]}] If a Steiner vertex $u\in S_i$ has real incident edges $(w,u)$ and $(u,v)$ for $v,w\in S_{<i}$, then there is a path $P$ with exactly one negative edge containing only vertices in $S_{<i}$ from $w$ to $v$ of weight $\omega(w,u)+\omega(u,v)$. Furthermore, letting $(u',v)$ denote the final edge in $P$, we have $\omega(u',v)\ge \omega(u,v)$. In particular, this implies that the weight of the subpath of $P$ between $w$ and $u'$ is at most $\omega(w,u)$. Symmetrically, letting $(w,u')$ denote the first edge in $P$, we have $\omega(w,u')\ge \omega(w,u)$, which implies that the weight of the subpath of $P$ between $u'$ and $v$ is at most $\omega(u,v)$.
		\item [{\crtcrossreflabel{(I4)}[item:steiner-invariant-3]}] For any (real or auxiliary) negative edge $(r,r')$, both $r$ and $r'$ are in level 0.
	\end{enumerate}

	In our iterative shortcutting algorithm, we need to reweight the graph, after which the invariants may not hold anymore. We show the following lemma in \Cref{sec:maintain-invariant}, asserting that after any reweighting, the shortcut edges can be modified so that the invariants are maintained and the shortcutting progress is not affected.
	
	\begin{lemma}[Informal version of \Cref{lem:maintain-invariant-detailed}]\label{lem:maintain-invariant}
		Let $G=(V,E)$ be a directed graph with no Steiner vertices.
		We can maintain data structure $\textsc{MaintainInvariant}(G)$ with the following operations:
		\begin{enumerate}
			\item $\textsc{Insert}(\vec{\Delta})$: Create a new set of Steiner vertices $S_{t}$ defined by $\vec{\Delta}$. The data structure will add and maintain shortcut edges to and from $S_{t}$ required by the shortcutting algorithm. 
			
			\item $\textsc{Reweight}(\phi)$: Reweight the graph via a valid potential function $\phi$. The data structure will maintain Steiner edges so that they still satisfy invariants \
			\ref{item:steiner-invariant-1}--\ref{item:steiner-invariant-3}.
		\end{enumerate}
	\end{lemma}
    We will assume that our input graphs come with this data structure initialized. For the original input graph, this is just an empty data structure. For the input graphs to recursive calls on a different graph, we will always initialize this data structure before the recursive call.

    In our proofs in \Cref{sec:unfolding}, we will need to simultaneously reason about graphs from different iterations. For this, we introduce some notation. Let $G(t)$ denote the graph after $t$ iterations of shortcutting, after applying the reweighting but before applying the data structure. Let $\phi_t$ denote the (betweenness reduction) reweighting which we apply at iteration $t$ and define $\phi_{(i,j]}(u)=\sum_{t\in(i,j]}\phi_t(u)$ to be the total reweighting applied after iteration $i$. Finally, let $i(u)$ for vertex $u$ denote the iteration such that $u\in S_{i(u)}$.
        
		
		

	\section{Shortcutting Algorithm}
	
	As outlined in the technical overview, our algorithm is an iterative shortcutting algorithm, which proceeds over $\Theta(\log{k})$ iterations. In each iteration, we add some nodes and edges to the graph in order to decrease the number of negative edges on shortest paths. Let $G_0=(V_0,E_0)$ be the original input graph, with $m_0=|E_0|$ and $n_0=|V_0|$.
    
	

    \begin{theorem}\label{thm:one-iteration-shortcutting}
		Let $G=(V,E)$ be a (directed, weighted) graph with $m$ edges and $k$ negative vertices satisfying invariants \ref{item:negative-edge-invariant-1}--\ref{item:steiner-invariant-3}. There is a randomized algorithm that outputs a supergraph $G'=(V',E')$ with $V'\supseteq V$ and $E'\supseteq E$, and a potential $\phi$ on $V$ such that
		\begin{enumerate}
			\item For any $s,t\in V$, we have $d_{G'}(s,t)=d_{G}(s,t)+\phi(s)-\phi(t)$,\label{item:one-iteration-shortcutting-1}
			\item For any $s,t\in V$ and hop bound $h$, we have $d_{G'}^{h-\lfloor h/3\rfloor}(s,t)\le d_{G}^h(s,t)+\phi(s)-\phi(t)$,
			\item $V'$ has at most $O(k)$ additional Steiner vertices, so $|V'|\le |V|+O(k)$, 
			\item $E'$ has at most $|E_0|^{1+o(1)}$ additional edges\footnote{In recursive instances, the size of the graph increases by a factor of $\operatorname{polylog}(n)$, so the size of the corresponding ``input graph'' to those recursive instance increase by the same factor. Since the recursion depth is at most $O(\sqrt{\log{n}})$ in our algorithm, this increase is by a factor of at most $n^{o(1)}$ in total, so the sizes of the input graph for the recursive instances and the original input graph are the same, up to sub-polynomial factors.}, so $|E'|\le |E|+|E_0|^{1+o(1)}$,\label{item:one-iteration-shortcutting-3}
			\item There are no new negative vertices.
		\end{enumerate}
		The algorithm uses one call to negative weight shortest paths on $\tilde{O}(|E|)$ edges and $O(k/2^{\sqrt{\log{n}}})$ negative vertices, and takes $m^{1+o(1)}$ additional time.
	\end{theorem}
	
	Before we prove \Cref{thm:one-iteration-shortcutting}, we first show how to apply it to prove \Cref{thm:main}.
	
	\begin{proof}[Proof of \Cref{thm:main}]
		If $k\le 2^{\sqrt{\log{n}}}$ at any point, then naively computing $k$-hop shortest paths suffices since the running time $\tilde O(mk)=m^{1+o(1)}$ meets the desired bound. Otherwise, each time we apply \Cref{thm:one-iteration-shortcutting}, the number of negative hops along shortest paths drops by a constant factor. By applying the shortcutting procedure $O(\log{k})\le O(\log{n})$ times, we guarantee that all shortest paths can be replaced with a path that has at most 2 negative edges. On the resulting graph, we can run $2$-hop shortest paths to compute the desired single-source distances in the final graph in $m^{1+o(1)}$ time. Using the reweightings $\phi$ and property~(\ref{item:one-iteration-shortcutting-1}), we can then recover the shortest paths in the original graph.
		
		By property~(\ref{item:one-iteration-shortcutting-3}), each iteration increases the number of edges by $|E_0|^{1+o(1)}$, so all graphs have at most $m+|E_0|^{1+o(1)}$ edges. Thus, the runtime is dominated by $O(\log{n})$ recursive calls to single-source shortest paths, and $m^{1+o(1)}$ additional time, giving the following recursion:
		\begin{align*}
			T(m,k)\le O(\log{n})\cdot T\left(m+|E_0|^{1+o(1)},k/2^{\sqrt{\log{n}}}\right)+m^{1+o(1)}.
		\end{align*}
		Applying this on the original graph $G_0$, there are at most $\sqrt{\log{n}}$ levels of recursion since the number of negative vertices drops by $2^{\sqrt{\log{n}}}$ in each recursive call. On recursion level $i\ge 1$, there are $O(\log n)^i$ recursive instances, each with at most $|E_0|^{1+o(1)}$ edges. Since $i\le \sqrt{\log{n}}$, the total size of all instances is at most $|E_0|^{1+o(1)}$ so the total runtime is still $|E_0|^{1+o(1)}$.
	\end{proof}

	We now give our algorithm for \Cref{thm:one-iteration-shortcutting}. The high level approach is the same as in \cite{li2025shortcutting}. First, we apply a reweighting to the graph to reduce the betweenness for every pair of vertices. In the reweighted graph, we compute forward and backward searches from each negative vertex such that every vertex in the backward search can negatively reach every vertex in the forward search using few hops. Finally, we add shortcut edges between the backward and forward searches for each negative vertex to decrease the number of hops needed in a shortest path by a constant factor. To implement this strategy more efficiently than \cite{li2025shortcutting}, we need to strengthen each of the steps.

	\subsection{Betweenness Reduction}\label{sec:recursive-betweenness}
	
	For betweenness reduction, we use the multi-hop betweenness reduction algorithm in \cite{li2025shortcutting}, which constructs a layered graph, and reduces the problem to a shortest path problem on this layered graph. We show that this version of betweenness reduction has the following properties. First, it fixes a subset of negative vertices, so we can apply \Cref{lem:betweenness-weighted} to get an almost-linear bound on search sizes. In addition, all of our invariants on negative edges and Steiner vertices still hold in the layered graph, so the structure of Steiner vertices is well maintained in recursions.
	
	\begin{lemma}\label{lem:betweenness-reduction-weighted}
		There is an algorithm that, given a graph $G$ with vertex partition $S_0\sqcup S_1\sqcup\cdots\sqcup S_t$ that satisfies invariants~\ref{item:negative-edge-invariant-1}--\ref{item:steiner-invariant-3}, a subset $U\subset N$ of its negative vertices, and a hop bound $h$, either reports a negative cycle in $G$ or finds a valid reweighting $\phi$ such that after reweighting,
		\[\tilde d^h(s,r)+\tilde d^h(r,t)\ge 0 \quad \forall r\in U,\ \forall s,t\in V.\]
		The algorithm computes this reweighting by solving the single-source shortest path problem on an auxiliary graph $G'$ with vertex partition $S'_0\sqcup S'_1\sqcup\cdots\sqcup S'_t$, which satisfies invariants~\ref{item:negative-edge-invariant-1} and \ref{item:steiner-invariant-1}-\ref{item:steiner-invariant-3}. The auxiliary graph $G'$ contains $O(h)$ copies of the original graph $G$, along with the data structure for maintaining invariants in \Cref{lem:maintain-invariant}. All negative edges in these copies are treated as auxiliary edges. In addition, the graph has $O(|V(G)|)$ extra vertices, $O(h|E(G)|)$ extra edges, and $|U|$ real negative edges. 
	\end{lemma}
	
	\begin{proof}
		The auxiliary graph $G'$ is constructed as follows. Let $G^+$ denote the graph $G$ where all negative edges are considered as auxiliary. We start with $2h+1$ copies of $G^+$, denoted as $G_0$, $G_1^{\textup{forward}},\ldots,G_h^{\textup{forward}}$, and $G_1^{\textup{backward}},\ldots,G_h^{\textup{backward}}$. We use the name of the vertex with the subscript of a graph copy, such as $v_1^{\textup{forward}}$, to refer the vertex in that copy. For each vertex $v\in S_i$, all copies of $v$ belong to $S'_i$. That is, the copies of original vertices are also original in $G'$, and the copies of Steiner vertices retain their level in $G'$.

		{Since we maintain some of the shortcut edges in $G$ by the data structure in \Cref{lem:maintain-invariant}, we initialize such a data structure for $G'$ to maintain the shortcut edges in each copy of $G^+$ in further recursions, as follows. For each copy of $G^+$ in $G'$, we make an independent copy of the data structure on $G$ that works on this copy of $G^+$, which can function normally in $G'$ since we keep all negative edges as auxiliary edges. These together initialize the data structure for $G'$.}
		
		Next, we add edges between these copies as follows; see Figure~\ref{fig:betweenness-graph} for a visual reference. Let $M\ge0$ be a value larger than the absolute value of any (real or auxiliary) negative edge weight. For every real negative or auxiliary negative edge $(u,v)$ in $G$, we add the following edges of weight $M+\omega(u,v)\ge0$:
		\begin{itemize}
			\item An edge from $u_0$ to $v_1^{\textup{forward}}$.
			\item An edge from $u_i^{\textup{forward}}$ to $v_{i+1}^{\textup{forward}}$ for every $i\in [h-1]$.
			\item An edge from $u_{i+1}^{\textup{backward}}$ to $v_i^{\textup{backward}}$ for every $i\in [h-1]$.
			\item An edge from $u_1^{\textup{backward}}$ to $v_0$.
		\end{itemize}
		Next, for every original vertex $v\in S_0$, we add the following ``self edges'' of weight $M$:
		\begin{itemize}
			\item An edge from $v_0$ to $v_1^{\textup{forward}}$.
			\item An edge from $v_i^{\textup{forward}}$ to $v_{i+1}^{\textup{forward}}$ for every $i\in [h-1]$.
			\item An edge from $v_{i+1}^{\textup{backward}}$ to $v_i^{\textup{backward}}$ for every $i\in [h-1]$.
			\item An edge from $v_1^{\textup{backward}}$ to $v_0$.
		\end{itemize}
		Finally, for every vertex $v\in U$, we add two vertices $v^{\textup{tail}}$ and $v^{\textup{head}}$, which are also considered original vertices in $G'$. We add zero-weight edges from $v_h^{\textup{forward}}$ to $v^{\textup{tail}}$, and from $v^{\textup{head}}$ to $v_h^{\textup{backward}}$. Moreover, we add a real negative edge of weight $-2hM$ from $v^{\textup{tail}}$ to $v^{\textup{head}}$. These are the only negative edges in $G'$. This completes the construction of $G'$. From this construction, $G'$ clearly satisfies invariant \ref{item:negative-edge-invariant-1} since the only edges adjacent to $v^{\textup{tail}}$ and $v^{\textup{head}}$ are the three types of edges above.
		
		\begin{figure}
			\centering
			\includegraphics{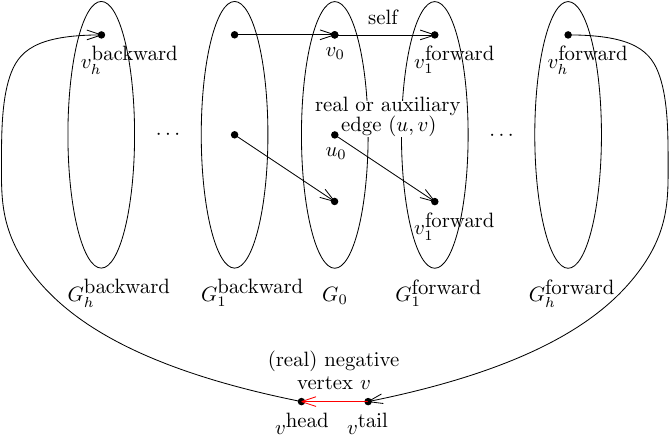}
			\caption{The auxiliary graph $G'$. Only the red edges $(v^{\textup{tail}},v^{\textup{head}})$ can be negative.}
			\label{fig:betweenness-graph}
		\end{figure}
		
		The algorithm computes single-source shortest path on graph $G'$. If no negative cycle is found, then let $\phi$ be the resulting distance function, which neutralizes the real negative edges in $G'$. The algorithm then sets $\phi(v_0)$ as the potential for $v$ in $G$, and returns this potential.
		
		We now prove the correctness of this algorithm. First, suppose that the algorithm finds a potential $\phi$. Since $G_0$ is a copy of the non-negative edges of $G$, taking $\phi(v_0)$ as $v$'s potential gives a valid potential in $G$. To show that this potential gives $\tilde d^h(s,r)+\tilde d^h(r,t)\ge 0$ for every $r\in U$ and $s,t\in V(G)$ after reweighting, we prove the following.
		\begin{subclaim}\label{subcl:layered-graph-dist-bw}
			For a valid potential $\phi$ neutralizing all negative edges in $G'$ and any $s,t\in V(G)$, we have
			\[\phi(t_0)-\phi(s_0)\le \min_{r\in U}\tilde d_G^h(s,r)+\tilde d_G^h(r,t).\]
		\end{subclaim}
		\begin{subproof}
			We claim that it suffices to find a path from $s_0$ to $t_0$ with total weight $\displaystyle\min_{r\in U}\tilde d_G^h(s,r)+\tilde d_G^h(r,t)$ in $G'$. To see this claim, observe that the weight of this path is $\displaystyle\min_{r\in U}\tilde d_G^h(s,r)+\tilde d_G^h(r,t)+\phi(s_0)-\phi(t_0)$ after reweighting by $\phi$. This new weight is also non-negative since all edges are non-negative after reweighting, so we conclude that $\displaystyle\min_{r\in U}\tilde d_G^h(s,r)+\tilde d_G^h(r,t)+\phi(s_0)-\phi(t_0)\ge0$. For the rest of the proof, we establish this path from $s_0$ to $t_0$ in $G'$.
			
			Let $r\in U$ be the vertex minimizing $\tilde d_G^h(s,r)+\tilde d_G^h(r,t)$. By definition of $\tilde d_G^h(s,r)$, there is a path from $s$ to $r$ in $G$ with weight $\tilde d_G^h(s,r)$ that has $i\le h$ negative edges. We trace this path in $G'$ by starting at $s_0$ in $G_0$ and ending at $r_i^{\textup{forward}}$ in $G_i^{\textup{forward}}$, taking the edge of weight $M+\omega(u,v)$ to the next copy whenever a negative edge $(u,v)$ is encountered. Since $r$ is a negative vertex and thus an original vertex, we can append self edges from $r_i^{\textup{forward}}$ to get the path $s_0\leadsto r_i^{\textup{forward}}\leadsto r_h^{\textup{forward}}$, whose total weight is $\tilde d_G^h(s,r)+hM$. Similarly, there is a path of weight $\tilde d_G^h(r,t)+hM$ from $r_h^{\textup{backward}}$ to $t_0$. By combining these two paths with negative edge $(r^{\textup{tail}}, r^{\textup{head}})$, we get a path from $s_0$ to $t_0$ with total weight $\tilde d_G^h(s,r)+\tilde d_G^h(r,t)$, as promised.
		\end{subproof}
		
		
		Second, suppose that the algorithm detects a negative cycle in $G'$. We claim that $G$ also contains a negative cycle.
		\begin{subclaim}
			For a negative cycle in $G'$, we can construct a closed walk in $G$ with the same total weight.
		\end{subclaim}
		\begin{subproof}
			Since edges inside a copy are non-negative and edges between copies form the \[G_0\rightarrow G_1^{\textup{forward}}\rightarrow\cdots\rightarrow G_h^{\textup{forward}}\rightarrow G_h^{\textup{backward}}\rightarrow\cdots\rightarrow G_1^{\textup{backward}}\rightarrow G_0\] cycle structure, any negative cycle must visit $G_0$ at least once. We split the cycle into subpaths that start and end at $G_0$, and which cycle through the copies exactly once. In particular, each subpath uses exactly one negative edge $(r^{\textup{tail}},r^{\textup{head}})$ when going from $G_h^{\textup{forward}}$ to $G_h^{\textup{backward}}$.
			We now show that each subpath corresponds to a walk in $G$ with the same weight. This implies that the whole negative cycle corresponds to a closed walk in $G$ with the same weight, as promised. Assume that the negative edge used in this path is $(r^{\textup{tail}},r^{\textup{head}})$, and the path is of the form $s_0\leadsto r_h^{\textup{forward}}\to r^{\textup{tail}}\to r^{\textup{head}}\to r_h^{\textup{backward}}\leadsto t_0$. By mapping each vertex $v_i^{\textup{forward}}$ to the original vertex $v$, the subpath $s_0\leadsto r_h^{\textup{forward}}$ in $G'$ corresponds to a walk $s\leadsto r$ in $G$ whose weight is lower by exactly $hM$. Similarly, the subpath $r_h^{\textup{backward}}\leadsto t_0$ corresponds to a walk $r\leadsto t$ in $G$ whose weight is lower by exactly $hM$. Finally, the negative edge $(r^{\textup{tail}},r^{\textup{head}})$ in $G'$ has weight $-2hM$, so the total weight of the subpath $s_0\leadsto t_0$ in $G'$ and the walk $s\leadsto t$ in $G$ are equal.
		\end{subproof}
		
		Let $n=|V(G)|$. We now prove the size guarantees of $G'$. We make $2h+1$ copies of the original graph, with all negative edges treated as auxiliary. Since $G$ satisfies invariant \ref{item:negative-edge-invariant-1}, it contains at most $n$ (real or auxiliary) negative edges, so the number of additional vertices and edges between layers is $O(nh)$. The only real negative edges are $(r^{\textup{tail}},r^{\textup{head}})$ for each $r\in U$, so there are $|U|$ many negative edges.

		Finally, we prove invariants \ref{item:steiner-invariant-1}-\ref{item:steiner-invariant-3} hold in $G'$ assuming that they hold for $G$.
		Since we make full copies of the graph $G$ along with the data structure on $G$, these invariants hold within each copy. For edges between different copies, the observation is that they only connect between original copies.
		
		\begin{subclaim}
			Suppose that an edge in $G'$ is not contained in a copy of the original graph, then both terminals of this edge are either original vertices on a copy of the original graph, or are vertices $r^{\textup{tail}},r^{\textup{head}}$ for some $r\in U$, which are also original vertices.
		\end{subclaim}
		
		\begin{subproof}
			We examine all edges between different layers in $G'$. First, for every negative edge $(u,v)$ in $G$, we add edges from a copy of $u$ to a copy of $v$. Since we assume that $G$ satisfies invariant \ref{item:steiner-invariant-3}, these vertices are original vertices.
			Second, for every original vertex $u$, we add self edges between copies of $u$, which are naturally between original vertices.
			Finally, for every $r\in U$, we add a path $r_h^{\textup{forward}}\to r^{\textup{tail}}\to r^{\textup{head}}\to r_h^{\textup{backward}}$. Since $G$ satisfies invariant \ref{item:steiner-invariant-3}, all negative vertices are original vertices, so these edges are still between original vertices and $r^{\textup{tail}},r^{\textup{head}}$.
		\end{subproof}
		
		Since edges between copies are all between original vertices, invariant \ref{item:steiner-invariant-1}-\ref{item:steiner-invariant-3} holds for them since these invariants are only about Steiner vertices.
    \end{proof}

	\subsection{Forward and Backward Searches}\label{sec:forward-backward-search}
	
	We now provide a weighted version of the algorithm in \cite{li2025shortcutting} that computes forward and backward searches from each vertex in $N$.
	To illustrate the idea, we first prove a $1$-hop version of the result, and then give the full $h$-hop version by reducing the problem to the $1$-hop version.
	
	\begin{lemma}\label{lem:forward-backward-search-weighted}
		Consider a graph $G$ that satisfies invariant \ref{item:negative-edge-invariant-1}, with non-negative weight $\lambda(v)$ on every vertex $v$. Let $\Lambda=\sum_{v\in V}\lambda(v)$ denote the total weight. There is an algorithm that, for every negative vertex $r$, computes $\Delta_r$ with the following:
		\begin{itemize}
			\item A set $V_r^{\textup{in}}=\{x\in V:d^0(x,r)<\Delta_r\}$ along with the distance $d^0(x,r)$ for every $x$ in $V_r^{\textup{in}}$.
			\item A set $V_r^{\textup{out}}=\{x\in V:d^1(r,x)<-\Delta_r\}$ along with the distance $d^1(r,x)$ for every $x$ in $V_r^{\textup{out}}$.
		\end{itemize}
		Moreover, the algorithm also has the following guarantees. Let $U\subseteq N$ be a random sample where each vertex in $N$ is sampled with probability $1/b$. With high probability over the randomness of $U$, the following holds. For any valid potential function $\phi$ such that 
		\[d^0_{\phi}(s,r)+d^1_{\phi}(r,t)\ge 0\quad\forall r\in U,\ \forall s,t\in V,\] 
		the algorithm on the reweighted graph $G_\phi$ satisfies
		\[ \sum_{r\in N}\Big(\sum_{v\in V_r^{\textup{in}}}\lambda(v)+\sum_{v\in V_r^{\textup{out}}}\lambda(v)\Big)\le\tilde{O}(\Lambda b) ,\]
		and the algorithm runs in $\tilde O(mb)$ time.
	\end{lemma}
	\begin{proof}
		To bound the running time, we define $\lambda'(v)=\lambda(v)+\frac{\operatorname{deg}(v)}{2m}\cdot \Lambda$ and mainly use $\lambda'(v)$ in the algorithm. The extra term $\frac{\operatorname{deg}(v)}{2m}\cdot \Lambda$ is used to bound the running time of searches, which is related to the total volume of $V_r^{\textup{in}}$ and $V_r^{\textup{out}}$. In addition, $\sum_v \lambda'(v)=\Lambda+\sum_v \lambda(v)=2\Lambda$, so this change only increases the total weight by a constant factor.
		
		For each negative edge $(r,r')$, the algorithm runs two Dijkstra's algorithms in parallel, called the forward search and the backward search. The forward search runs on all non-negative edges of $G$, and starts from $r'$; the backward search runs on all non-negative edges but with edge directions reversed, and starts from $r$. For the forward search, Dijkstra's algorithm maintains a list $P^{\textup{out}}$ of processed vertices, and the minimum distance $d^{\textup{out}}$ to the next unprocessed vertex. Similarly, the algorithm maintains $P^{\textup{in}}$ and $d^{\textup{in}}$ for the backward search. In each step, the algorithm chooses one search and processes the next vertex in this search. The way of choosing the search is as follows:
		\begin{itemize}
			\item If $d^{\textup{in}}+d^{\textup{out}}\ge -\omega(r,r')$, set $\Delta_r=d^{\textup{in}}$ and immediately terminate. Let $V_r^{\textup{in}}$ be the set of vertices $x$ with $d^0(x,r)<\Delta_r$ in the backward search, which is a subset of $P^{\textup{in}}$ since the vertices outside $P^{\textup{in}}$ have distance at least $d^{\textup{in}}=\Delta_r$ to $r$. Similarly, let $V_r^{\textup{out}}$ be the set of vertices $x$ with $d^0(r',x)<-\Delta_r-\omega(r,r')$ in the forward search, which is a subset of $P^{\textup{out}}$ since the vertices outside $P^{\textup{out}}$ have distance at least $d^{\textup{out}}\ge-d^{\textup{in}}-\omega(r,r')=-\Delta_r-\omega(r,r')$ from $r'$. Therefore, the algorithm can find $V_r^{\textup{in}}$ and $V_r^{\textup{out}}$ with their distances by checking the processed vertices $P^{\textup{in}}$ and $P^{\textup{out}}$.
			\item Otherwise, compare the total weight of processed vertices, plus the next vertex to be processed, in both searches, and choose the smaller side to continue. More formally, let $u^{\textup{out}}$ denote the next vertex to be explored in the forward search, and $u^{\textup{in}}$ denote the next vertex to be explored in the backward search. The algorithm chooses the forward search if and only if
			\[\sum_{x\in P^{\textup{out}}\cup \{u^{\textup{out}}\}} \lambda'(x)\le \sum_{x\in P^{\textup{in}}\cup \{u^{\textup{in}}\}} \lambda'(x).\]
		\end{itemize}
		We first prove the weighted size bound on processed vertices in searches. To connect the total weight of both searches to the sum of the smaller side in \Cref{lem:betweenness-weighted}, we have the following observation.
		
		\begin{subclaim}
			After each step in the above algorithm, we can find a $\Delta_r'$ that satisfies the following: Let $(V_r^{\textup{out}})'=\{x\in V:d^1(r,x)<-\Delta_r'\}$, $(V_r^{\textup{in}})'=\{x\in V:d^0(x,r)<\Delta_r'\}$; the total weight of processed vertices $P^{\textup{out}}$ and $P^{\textup{in}}$ can be bounded by 
			\[\max\Big\{\sum_{x\in P^{\textup{out}}}\lambda'(x),\sum_{x\in P^{\textup{in}}}\lambda'(x)\Big\}\le \min\Big\{\sum_{x\in (V_r^{\textup{out}})'}\lambda'(x),\sum_{x\in (V_r^{\textup{in}})'}\lambda'(x)\Big\}.\]
		\end{subclaim}
		\begin{subproof}
			We induct on the steps of the algorithm. The base case of no vertices processed is trivial. Next, assuming that all previous steps satisfy the induction hypothesis, we consider the next step performed by the algorithm. Since the algorithm does not terminate, we have $d^{\textup{out}}+d^{\textup{in}}<-\omega(r,r')$. Let $\epsilon>0$ be small enough that $d^{\textup{out}}+d^{\textup{in}}+\epsilon<-\omega(r,r')$; we choose $\Delta_r'=d^{\textup{in}}+\epsilon$ for the statement in the claim. Since $\Delta_r'>d^{\textup{in}}$, we have $P^{\textup{in}}\cup\{u^{\textup{in}}\}\subset (V_r^{\textup{in}})'$ because the latter subset contains all vertices with distance less than $\Delta_r'$. Similarly, we claim that $P^{\textup{out}}\cup\{u^{\textup{out}}\}\subset (V_r^{\textup{out}})'$: since each vertex $x\in P^{\textup{out}}\cup\{u^{\textup{out}}\}$ satisfies $d^0(r',x)\le d^{\textup{out}}<-\Delta'_r-\omega(r,r')$, we have $d^1(r,x)\le\omega(r,r')+d^0(r',x)<-\Delta'_r$, so $x\in(V_r^{\textup{out}})'$. Therefore,
			\[\min\Big\{\sum_{x\in P^{\textup{out}}\cup \{u^{\textup{out}}\}}\lambda'(x),\sum_{x\in P^{\textup{in}}\cup \{u^{\textup{in}}\}}\lambda'(x)\Big\}\le \min\Big\{\sum_{x\in (V_r^{\textup{out}})'}\lambda'(x),\sum_{x\in (V_r^{\textup{in}})'}\lambda'(x)\Big\}.\]
			Now, we can establish the inductive statement after this step. Without loss of generality, suppose that $\sum_{x\in P^{\textup{out}}\cup \{u^{\textup{out}}\}} \lambda'(x)\le \sum_{x\in P^{\textup{in}}\cup \{u^{\textup{in}}\}} \lambda'(x)$ and the algorithm chooses the forward search. We prove this bound by considering two cases.
			\begin{itemize}
				\item Suppose that $\sum_{x\in P^{\textup{out}}\cup \{u^{\textup{out}}\}} \lambda'(x)\ge \sum_{x\in P^{\textup{in}}} \lambda'(x)$. Then, using the previous bound, we are already done.
				\item Otherwise, $\sum_{x\in P^{\textup{in}}} \lambda'(x)$ is the larger side. In this case, we can use the induction hypothesis on the previous step to bound $\sum_{x\in P^{\textup{in}}} \lambda'(x)$, using the value of $\Delta'_r$ from that previous step.
			\end{itemize}
			Therefore, in both cases, we can find a $\Delta_r'$ to complete the induction and the proof.
		\end{subproof}
		
		By applying \Cref{lem:betweenness-weighted} with $\lambda'(v)$ as the vertex weight, the sets $(V_r^{\textup{out}})'$ and $(V_r^{\textup{in}})'$ found in the above subclaim satisfy
		\[\sum_r\min\Big\{\sum_{x\in (V_r^{\textup{out}})'}\lambda'(x),\sum_{x\in (V_r^{\textup{in}})'}\lambda'(x)\Big\}\le \tilde O(\Lambda b).\]
		Therefore, the subclaim shows that
		\[\sum_r\Big(\sum_{x\in P^{\textup{out}}_r}\lambda'(x)+\sum_{x\in P^{\textup{in}}_r}\lambda'(x)\Big)\le \tilde O(\Lambda b),\]
		where $P^{\textup{out}}_r$ and $P^{\textup{in}}_r$ denote the processed vertices in the forward and backward search from negative edge $(r,r')$. Recall that $\lambda'(v)=\lambda(v)+\frac{\operatorname{deg}(v)}{2m}\cdot \Lambda$, so both the sum of $\lambda(v)$ and the sum of $\frac{\operatorname{deg}(v)}{2m}\cdot \Lambda$ are bounded by $\tilde O(\Lambda b)$. The former one implies the statement in \Cref{lem:forward-backward-search-weighted} because $V_r^{\textup{out}}$ is a subset of $P_r^{\textup{out}}$ and $V_r^{\textup{in}}$ is a subset of $P_r^{\textup{in}}$; the latter one implies that the total volume of all searches is $\tilde O(mb)$, so the total running time of the forward and backward searches is $\tilde O(mb)$.
	\end{proof}
	
	
	For the multi-hop case, we reduce it to the $1$-hop case above, using the idea of the layered graph in \cite{li2026bellmanfordalmostlineartimedense}. The main difference between the multi-hop case and the $1$-hop case is that we now force the searches $V_r^{\textup{in}}$ and $V_r^{\textup{out}}$ to include all intermediate vertices along an $h$-hop path of small enough weight, even if the vertices themselves have large $h$-hop distance.
	
	\begin{lemma}\label{lem:multi-hop-forward-backward-search-weighted}
		Consider a graph $G$ (with auxiliary edges $E_{\textup{fake}}$) that satisfies invariant \ref{item:negative-edge-invariant-1}, with non-negative weight $\lambda(v)$ on every vertex $v$, and a hop bound $h$. Let $\Lambda=\sum_{v\in V}\lambda(v)$ be the total weight. There is an algorithm that, for every negative vertex $r$, computes $\Delta_r$ with the following sets: 
		\begin{itemize}
			\item A set $\tilde V_r^{\textup{in}}$ that contains all vertices on any $h$-hop path (which may use auxiliary edges) that ends at $r$ and whose total weight is less than $\Delta_r$.
			\item A set $\tilde V_r^{\textup{out}}$ that contains all vertices on any $(h+1)$-hop path (which may use auxiliary edges) that begins at $r$ and whose total weight is less than $-\Delta_r$.
		\end{itemize}
		Moreover, the algorithm also has the following guarantees. Let $U\subseteq N$ be a random sample where each vertex in $N$ is sampled with probability $1/b$. With high probability over the randomness of $U$, the following holds. For any valid potential $\phi$ that satisfies
		\[\tilde d^h_{\phi}(s,r)+\tilde d^{h+1}_{\phi}(r,t)\ge 0,\quad\forall r\in U,\ \forall s,t\in V,\]
		the algorithm on the reweighted graph $G_\phi$ satisfies
		\[ \sum_{r\in N}\Big(\sum_{v\in V_r^{\textup{in}}}\lambda(v)+\sum_{v\in V_r^{\textup{out}}}\lambda(v)\Big)\le\tilde{O}(\Lambda hb) ,\]
		and the algorithm runs in $\tilde O(mhb)$ time.
	\end{lemma}
	
	\begin{proof}
		We construct a layered graph $H$ similar to \cite{li2026bellmanfordalmostlineartimedense} and run the algorithm in \Cref{lem:forward-backward-search-weighted} on it. To define this layered graph, we first compute $\tilde d^i(V,v)$ and $\tilde d^i(v,V)$ in $G$ for each $v\in V$ and $i\le h$, which takes $\tilde O(mh)$ time using the hybrid Dijkstra and Bellman-Ford algorithm~\cite{dinitz2017hybrid}. With these distances, we construct the layered graph as follows; see Figure~\ref{fig:layered-graph} for a visual reference. First, let $G^+$ denote the graph $G$ with all negative edges removed. We start with $2h+3$ copies of $G^+$, denoted as $G_0^{\textup{in}},G_1^{\textup{in}},\ldots,G_h^{\textup{in}}$ and $G_0^{\textup{out}},G_1^{\textup{out}},\ldots,G_{h+1}^{\textup{out}}$. We use $v_i^{\textup{in}}$ to denote the copy of $v$ in $G_i^{\textup{in}}$, and $v_i^{\textup{out}}$ to denote the copy of $v$ in $G_i^{\textup{out}}$. Between these copies, we add edges as follows. For each negative (both real and auxiliary) edge $(u,v)$, we add the following edges with weight $\omega(u,v)$.
		\begin{figure} \centering
			\includegraphics{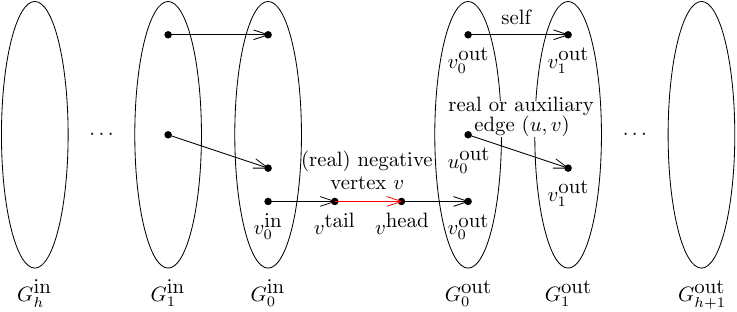}
			\caption{The layered graph $H$. Only the red edges $(v^{\textup{tail}},v^{\textup{head}})$ can be negative after reweighting by $\phi$.}\label{fig:layered-graph}
		\end{figure}
		\begin{itemize}
			\item An edge from $u_i^{\textup{out}}$ to $v_{i+1}^{\textup{out}}$, for every $i\in\{0,1,\ldots,h\}$.
			\item An edge from $u_{i+1}^{\textup{in}}$ to $v_i^{\textup{in}}$, for every $i\in\{0,1,\ldots,h-1\}$.
		\end{itemize}
		In addition, for each vertex $v$, we add the following ``self edges'' with weight $0$.
		\begin{itemize}
			\item An edge from $v_i^{\textup{out}}$ to $v_{i+1}^{\textup{out}}$, for every $i\in\{0,1,\ldots,h\}$.
			\item An edge from $v_{i+1}^{\textup{in}}$ to $v_i^{\textup{in}}$, for every $i\in\{0,1,\ldots,h-1\}$.
		\end{itemize}
		Next, for each negative vertex $v$, we add two vertices $v^{\textup{tail}}$ and $v^{\textup{head}}$ into $H$, with three zero-weight edges $(v_0^{\textup{in}},v^{\textup{tail}})$, $(v^{\textup{tail}},v^{\textup{head}})$, and $(v^{\textup{head}},v_0^{\textup{out}})$.
		Finally, we apply a weight function $\phi$ to reweight this graph, defined as follows. For vertices in $G_i^{\textup{out}}$, $\phi(v_i^{\textup{out}})=-\tilde d^{h+1-i}_G(v,V)$. For vertices in $G_i^{\textup{in}}$, $\phi(v_i^{\textup{in}})=\tilde d^{h-i}_G(V,v)$. Finally, $\phi(v^{\textup{tail}})=\phi(v_0^{\textup{in}})$ and $\phi(v^{\textup{head}})=\phi(v_0^{\textup{out}})$. We use $H$ to denote the graph reweighted with potential $\phi$. We note that this is not a valid reweighting (i.e., edge $(v^{\textup{tail}},v^{\textup{head}})$ might become a negative edge after reweighting). We emphasize that hops in $G$ are counted based on before the reweighting; in particular, $(v^{\textup{tail}},v^{\textup{head}})$ is not counted as a hop in $d^i_G(\cdot,\cdot)$.
		
		We first prove that, after reweighting, only the edge from $v^{\textup{tail}}$ to $v^{\textup{head}}$ can be negative in $H$. The non-negativity of other edges can be shown by the triangle inequality:
		\begin{itemize}
			\item For a non-negative edge $(u_i^{\textup{out}},v_i^{\textup{out}})$ in $G_i^{\textup{out}}$, its weight after reweighting is \[\omega_H(u_i^{\textup{out}},v_i^{\textup{out}})=\omega(u,v)-\tilde d^{h+1-i}_G(u,V)+\tilde d^{h+1-i}_G(v,V)\ge 0.\]
			\item For a non-negative edge $(u_i^{\textup{in}},v_i^{\textup{in}})$ in $G_i^{\textup{in}}$, its weight after reweighting is \[\omega_H(u_i^{\textup{in}},v_i^{\textup{in}})=\omega(u,v)+\tilde d^{h-i}_G(V,u)-\tilde d^{h-i}_G(V,v)\ge 0.\]
			\item For an inter-layer edge $(u_i^{\textup{out}},v_{i+1}^{\textup{out}})$ between $G_i^{\textup{out}}$ that corresponds to a negative edge $(u,v)$, its weight after reweighting is \[\omega_H(u_i^{\textup{out}},v_{i+1}^{\textup{out}})=\omega(u,v)-\tilde d^{h+1-i}_G(u,V)+\tilde d^{h-i}_G(v,V)\ge 0.\]
			\item For an inter-layer edge $(u_{i+1}^{\textup{in}},v_i^{\textup{in}})$ between $G_i^{\textup{in}}$ that corresponds to a negative edge $(u,v)$, its weight after reweighting is \[\omega_H(u_{i+1}^{\textup{in}},v_i^{\textup{in}})=\omega(u,v)+\tilde d^{h-i-1}_G(V,u)-\tilde d^{h-i}_G(V,v)\ge 0.\]
			\item For a self edge $(v_i^{\textup{out}},v_{i+1}^{\textup{out}})$ between $G_i^{\textup{out}}$, its weight after reweighting is \[\omega(v_i^{\textup{out}},v_{i+1}^{\textup{out}})=\tilde d^{h-i}_G(v,V)-\tilde d^{h+1-i}_G(v,V)\ge 0.\]
			\item For a self edge $(v_{i+1}^{\textup{in}},v_i^{\textup{in}})$ between $G_i^{\textup{in}}$, its weight after reweighting is \[\omega(v_{i+1}^{\textup{in}},v_i^{\textup{in}})=\tilde d^{h-i-1}_G(V,v)-\tilde d^{h-i}_G(V,v)\ge 0.\]
			\item Finally, the edge $(v_0^{\textup{in}},v^{\textup{tail}})$ and $(v^{\textup{head}},v_0^{\textup{out}})$ still have zero weights, since $\phi(v_0^{\textup{in}})=\phi(v^\textup{tail})$ and $\phi(v^{\textup{head}})=\phi(v_0^{\textup{out}})$, and $(v^{\textup{tail}},v^{\textup{head}})$ is negative after the reweighting. 
		\end{itemize}
		Therefore, the only possible negative edges in $H$ are the edges from $v^{\textup{tail}}$ to $v^{\textup{head}}$, so the graph satisfies invariant \ref{item:negative-edge-invariant-1}. We now show the key property of this layered graph $H$: the forward and backward search (defined in \Cref{lem:forward-backward-search-weighted}) starting from $r^{\textup{tail}}$ in $H$ is exactly the search we desired from $r$ in $G$.
		
		
		\begin{subclaim}\label{subcl:strong-bw-search-to-regular-search}
			For any vertex $u$ and any negative vertex $r$,
			\begin{itemize}
				\item $u$ is on an $(h+1)$-hop path of total weight less than $-\Delta_r$ that starts from $r$ if and only if there exists some $i\in\{0,1,\ldots,h+1\}$ such that $d_H(r^{\textup{tail}},u_i^{\textup{out}})<-\Delta_r+\phi(r^{\textup{tail}})$.
				\item $u$ is on an $h$-hop path of total weight less than $\Delta_r$ that ends at $r$ if and only if there exists some $i\in\{0,1,\ldots,h\}$ such that $d_H(u_i^{\textup{in}},r^{\textup{tail}})<\Delta_r-\phi(r^{\textup{tail}})$.
			\end{itemize}
		\end{subclaim}
		\begin{subproof}
			First, from the layered graph structure, the shortest path distances in $H$ satisfy the following.
			\begin{align}
				d_H(u_0^{\textup{out}},v_i^{\textup{out}})&=\tilde d^i_G(u,v)+\phi(u_0^{\textup{out}})-\phi(v_i^{\textup{out}})\label{eq:layered-graph-Vout}\\
				d_H(u_i^{\textup{in}},v_0^{\textup{in}})&=\tilde d^i_G(u,v)+\phi(u_i^{\textup{in}})-\phi(v_0^{\textup{in}})\label{eq:layered-graph-Vin}
			\end{align}
			
			Now we prove the subclaim. First, vertex $u$ is on an $(h+1)$-hop path of total weight less than $-\Delta_r$ that starts from $r$ if and only if there exists some $i\in\{0,1,\ldots,h+1\}$ such that $\tilde d^i_G(r,u)+\tilde d^{h+1-i}_G(u,V)<-\Delta_r$. On the other hand, using \eqref{eq:layered-graph-Vout} and the fact that edge $(r^{\textup{tail}},r^{\textup{head}})$ and $(r^{\textup{head}},r_0^{\textup{out}})$ both have weight zero before reweighting, we have
			\[
			d_H(r^{\textup{tail}},u_i^{\textup{out}})=\tilde d^i_G(r,u)+\phi(r^{\textup{tail}})-\phi(u_i^{\textup{out}})=\tilde d^i_G(r,u)+\tilde d^{h+1-i}_G(u,V)+\phi(r^{\textup{tail}}).
			\]
			Therefore, $\tilde d^i_G(r,u)+\tilde d^{h+1-i}_G(u,V)<-\Delta_r$ is equivalent to $d_H(r^{\textup{tail}},u_i^{\textup{out}})<-\Delta_r+\phi(r^{\textup{tail}})$.
			
			Similarly, for the $G_i^{\textup{in}}$ part, $u$ is on an $h$-hop path of total weight less than $\Delta_r$ that ends at $r$ if and only if there exists some $i\in\{0,1,\ldots,h\}$ such that $\tilde d^{h-i}_G(V,u)+\tilde d^i_G(u,r)<\Delta_r$. Because \eqref{eq:layered-graph-Vin} and edge $(r_0^{\textup{in}},r^{\textup{tail}})$ has weight zero, we get
			\[
			d_H(u_i^{\textup{in}},r^{\textup{tail}})=\tilde d^i_G(u,r)+\phi(u_i^{\textup{in}})-\phi(r^{\textup{tail}})=\tilde d^{h-i}_G(V,u)+\tilde d^i_G(u,r)-\phi(r^{\textup{tail}}).
			\]
			Therefore, $\tilde d^{h-i}_G(V,u)+\tilde d^i_G(u,r)<\Delta_r$ is equivalent to $d_H(u_i^{\textup{in}},r^{\textup{tail}})<\Delta_r-\phi(r^{\textup{tail}})$.
		\end{subproof}
		
		By this result, we can convert the constraint of being on an $h$-hop path of low total weight to regular distance constraints. In particular, it shows that an vertex $u$ is in $\tilde V_r^{\textup{out}}$ if and only if one of its copies $u_i^{\textup{out}}$ satisfies $d_H(r^{\textup{tail}},u_i^{\textup{out}})<-\Delta_r+\phi(r^{\textup{tail}})$; Similarly, $u$ is in $\tilde V_r^{\textup{in}}$ if and only if one of its copies $u_i^{\textup{in}}$ satisfies $d_H(u_i^{\textup{in}},r^{\textup{tail}})<\Delta_r-\phi(r^{\textup{tail}})$.
		
		Moreover, we can also convert the guarantee of $\tilde d^h_G(s,r)+\tilde d^{h+1}_G(r,t)\ge 0$ to a $1$-hop version in graph $H$. Recall that the statement says that the (reweighted) graph $G$ satisfies that 
		\[\tilde d^h_G(s,r)+\tilde d^{h+1}_G(r,t)\ge 0, \forall r\in U, \forall s,t\in V,\]
		so for any $u,v\in V$ and any $i\in\{0,1,\ldots,h\}$ and any $j\in\{0,1,\ldots,h+1\}$, we have
		\[
		\tilde d^{h-i}_G(s,u)+\tilde d^i_G(u,r)+\tilde d^j_G(r,v)+\tilde d^{h+1-j}_G(v,t)\ge 0, \forall r\in U, \forall s,t\in V.
		\]
		Therefore, if we consider any vertices $u_i^{\textup{in}}$ and $v_j^{\textup{out}}$ in $H$, then for any $r\in U$, we have
		\begin{align*}
			d_H(u_i^{\textup{in}},r^{\textup{tail}})+d_H(r^{\textup{tail}},v_j^{\textup{out}})&=\tilde d^i_G(u,r)+\tilde d^j_G(r,v)+\phi(u_i^{\textup{in}})-\phi(v_j^{\textup{out}})\\
			&=\tilde d^{h-i}_G(V,u)+\tilde d^i_G(u,r)+\tilde d^j_G(r,v)+\tilde d^{h+1-j}_G(v,V)\\
			&\ge 0,
		\end{align*}
		so the assumption of \Cref{lem:forward-backward-search-weighted} is satisfied. Therefore, we can apply it to compute these searches, while keep the total weighted search size bounded. To apply this lemma, for every vertex $v\in G$, we set the weight of all copies of $v$ in $H$, including $v_i^{\textup{out}}$, $v_i^{\textup{in}}$, $v^{\textup{tail}}$ and $v^{\textup{head}}$, as $\lambda(v)$. Since each vertex has $O(h)$ copies in $H$, the number of vertices in $H$ is $O(nh)$, and the total weight is $O(\Lambda h)$. For each negative vertex $r^{\textup{tail}}$, the lemma finds a $\Delta_{r^{\textup{tail}}}$, with the following two sets.
		\begin{itemize}
			\item $V_{r^{\textup{tail}}}^{\textup{out}}=\{x\in V(H):d_H^1(r^{\textup{tail}},x)<-\Delta_{r^{\textup{tail}}}\}$,
			\item $V_{r^{\textup{tail}}}^{\textup{in}}=\{x\in V(H):d_H^0(x,r^{\textup{tail}})<\Delta_{r^{\textup{tail}}}\}$.
		\end{itemize}
		In addition, the lemma ensures that the total weighted search size
		\[\sum_{r^{\textup{tail}}}\Big(\sum_{x\in V_{r^{\textup{tail}}}^{\textup{out}}}\lambda(x)+\sum_{x\in V_{r^{\textup{tail}}}^{\textup{in}}}\lambda(x)\Big)\]
		is $\tilde O(\Lambda hb)$, and the algorithm uses $\tilde O(mhb)$ time. 
		By the structure of the layered graph $H$ (see \Cref{fig:layered-graph}), any path from $r^{\textup{tail}}$ to other vertices uses at most one negative edge, and any path from other vertices to $r^{\textup{tail}}$ uses no negative edges. Therefore, the sets $V_{r^{\textup{tail}}}^{\textup{out}}$ and $V_{r^{\textup{tail}}}^{\textup{in}}$ satisfies that
		\begin{align*}
			V_{r^{\textup{tail}}}^{\textup{out}}&=\{x\in V(H):d_H(r^{\textup{tail}},x)<-\Delta_{r^{\textup{tail}}}\},\\
			V_{r^{\textup{tail}}}^{\textup{in}}&=\{x\in V(H):d_H(x,r^{\textup{tail}})<\Delta_{r^{\textup{tail}}}\},
		\end{align*}
		so by \Cref{subcl:strong-bw-search-to-regular-search}, $\tilde V_r^{\textup{out}}$ is the set of vertices that have at least one copy in $V_{r^{\textup{tail}}}^{\textup{out}}$, and $\tilde V_r^{\textup{in}}$ is the set of vertices that have at least one copy in $V_{r^{\textup{tail}}}^{\textup{in}}$.
		In addition, this also shows that the total weight 
		\[\sum_{r}\Big(\sum_{x\in V_r^{\textup{out}}}\lambda(x)+\sum_{x\in V_r^{\textup{in}}}\lambda(x)\Big)\]
		is at most $\tilde O(\Lambda hb)$. This concludes the proof of \Cref{lem:multi-hop-forward-backward-search-weighted}.
	\end{proof}

	\subsection{The Shortcutting Algorithm}\label{sec:algorithm}
	Using these tools, we present our algorithm to shortcut $G$ by a constant factor. This algorithm is similar to that of \cite{li2025shortcutting}, but we include the full description and proof for completeness.

	\paragraph{Step 1: Betweenness Reduction.} First, we reduce the (weighted) total betweenness of the graph. For a vertex $v\in S_0$, let $\textup{count}^+(v)$ and $\textup{count}^-(v)$ be the number of times $v$ has been in the forward and backward search of some negative vertex in previous iterations of shortcutting, respectively. Note that this can be easily stored as counters throughout the algorithm and passed in for recursive calls. Let $h=\Theta(\log^{10}{n})$; we use weight function $\lambda:V\to\mathbb{Z}_{\ge0}$ defined as
	\[ \lambda(v)=\begin{cases}\displaystyle
		\frac{m_0hb+nhb^2}{n}+\tilde O(hb)\cdot \deg_{G[S_0]}(v)+\sum_{w\in N^{\textup{out}}(v)\cap S_0}\textup{count}^-(w)+\sum_{u\in N^{\textup{in}}(v)\cap S_0}\textup{count}^+(u)&\text{if }v\in S_0\\
		0&\text{if }v\notin S_0
	\end{cases} \]
	We will apply the betweenness reduction algorithm (\Cref{lem:betweenness-reduction-weighted}) with hop parameter $h+1$ to compute reweighting $\phi$. We apply the reweighting to $G$ and also update the data structure by calling $\textsc{Reweight}(\phi)$. By the weighted version of our betweenness reduction lemma (\Cref{lem:betweenness-weighted}), this gives a weighted betweenness guarantee in $G_{\phi}$ with respect to weight function $\lambda$.

	\paragraph{Step 2: Forward/Backward Searches.} The second step is to compute (via \Cref{lem:multi-hop-forward-backward-search-weighted}) the forward and backward searches for each $r\in N$ with the same parameter $h=\Theta(\log^{10}{n})$ and weight function $\lambda$ on the graph $G_{\phi}$ with reduced betweenness. For each $r\in N$, this outputs a forward search $\tilde{V}_{r}^{\textup{out}}(r)$ and a backward search $\tilde{V}_r^{\textup{in}}(r)$, along with a scalar $\Delta_r$, such that
	\begin{align*}
		\tilde{V}^{\textup{out}}_r&\supseteq\{v\in V:d^{h+1}(r,v)<-\Delta_r\},\\
		\tilde{V}^{\textup{in}}_r&\supseteq\{v\in V:d^h(v,r)<\Delta_r\}.
	\end{align*}
	Let $\Delta_r'=\min\{-\omega(r,r'),\max\{0,\Delta_r\}\}$ be the value $\Delta_r$ clamped to the range $[0,-\omega(r,r')]$; we define $V^{\textup{out}}_r$ and $V^{\textup{in}}_r$ as follows.
	\begin{align*}
		V^{\textup{out}}_r&=\{v\in V:d^1(r,v)<-\Delta_r'\},\\
		V^{\textup{in}}_r&=\{v\in V:d^0(v,r)<\Delta_r'\}.
	\end{align*}
	First, observe that any $1$-hop path starting from $r$ is of weight at least $\omega(r,r')$, so $V^{\textup{out}}_r$ is empty if $\Delta_r\ge -\omega(r,r')$. Therefore, for any $v\in V^{\textup{out}}_r$, we have $d^1(r,v)<-\Delta_r$ because $\Delta_r<-\omega(r,r')$ when $v$ exists. Similarly, any $0$-hop path has a non-negative total weight, so $V^{\textup{in}}_r$ is empty if $\Delta_r<0$, and $\Delta_r'\le \Delta_r$ otherwise. Therefore, for any $v\in V^{\textup{in}}_r$, we have $d^0(v,r)<\Delta_r$. In addition, we have that $d^1(r,v)\ge d^{h+1}(r,v)$ and $d^0(v,r)\ge d^h(v,r)$. These together imply that $V^{\textup{out}}_r\subseteq \tilde{V}_r^{\textup{out}}$ and $V^{\textup{in}}_r\subseteq \tilde{V}_r^{\textup{in}}$. 
	In particular, the bounds on $\tilde{V}_r^{\textup{in}}$ and $\tilde{V}_r^{\textup{out}}$ which follow from \Cref{lem:multi-hop-forward-backward-search-weighted} also apply to $V^{\textup{in}}_r$ and $V^{\textup{out}}_r$.

	\paragraph{Step 3: Shortcutting.} Finally, we define the shortcut edges we add as follows. Edges in \ref{item:shortcut-step-2}--\ref{item:shortcut-step-5} are added as real edges, as they are used to shortcut the shortest paths, the same as in \cite{li2025shortcutting}. Edges in \ref{item:shortcut-step-6} are added as auxiliary edges, and their only goal is to ensure that the reweighted graph has some additional structural guarantees so that \Cref{lem:maintain-invariant} works, as we will show later in \Cref{sec:maintain-invariant}.
	\begin{itemize}
		\item[{\crtcrossreflabel{(S1)}[item:shortcut-step-1]}] For each negative vertex $r\in N$, create a new base vertex $\tilde{r}$, called a Steiner vertex.
		\item[{\crtcrossreflabel{(S2)}[item:shortcut-step-2]}] For each $r\in N$, $v\in V^{\textup{out}}_r\cup\{r'\}$, and $w\in N^{\textup{out}}(v)$, we add an edge $(\tilde{r},w)$ with weight $d^1(r,v)+\Delta_r'+\omega(v,w)$ if it is non-negative. 
		\item[{\crtcrossreflabel{(S3)}[item:shortcut-step-3]}] For each $r\in N$, $v\in V^{\textup{in}}_r\cup\{r\}$, and $u\in N^{\textup{in}}(v)$, we add an edge $(u,\tilde{r})$ with weight $\omega(u,v)+d^0(v,r)-\Delta_r'$ if it is non-negative. 
		\item[{\crtcrossreflabel{(S4)}[item:shortcut-step-4]}] For each negative vertex $r\in N$ and each $v\in V^{\textup{out}}_r$, if $v$ is an endpoint of a negative edge $(v,v')$, add an edge $(r,v')$ with weight $d^1(r,v)+\omega(v,v')$.
		\item[{\crtcrossreflabel{(S5)}[item:shortcut-step-5]}] For each negative vertex $r\in N$ and each $u'\in V^{\textup{in}}_r$, if $u'$ is the endpoint of a negative edge $(u,u')$, add an edge $(u,r')$ with weight $\omega(u,u')+d^0(u',r)+\omega(r,r')$.
		\item [{\crtcrossreflabel{(S6)}[item:shortcut-step-6]}] Finally, for each negative vertex $r\in N$, we add an edge $(\tilde r,r)$ with weight $\Delta_r'$, and an edge $(r',\tilde r)$ with weight $-\omega(r,r')-\Delta_r'$.
		In addition, if $\Delta_r'=0$, then we add a zero-weight edge $(r,\tilde r)$; if $\Delta_r'=-\omega(r,r')$, then we add a zero-weight edge $(\tilde r,r')$. In other words, if we clamp the value $\Delta_r$ because it is outside the range $(0,-\omega(r,r'))$, then we enforce $\Delta_r$ to be fixed at this value. We remind the reader that the edges added in this step are auxiliary edges, while edges added in previous steps are real edges.
	\end{itemize}
    Once we add these shortcut edges to the graph, we also add edges from \ref{item:shortcut-step-1}--\ref{item:shortcut-step-5} to the \textsc{MaintainInvariant} data structure. Finally, we apply \Cref{lem:one-negative-outgoing-edge} on the shortcutted graph to get the final $G'$ in order to ensure invariant \ref{item:negative-edge-invariant-1}.

    \subsection{Properties of Shortcutting Algorithm}

    In this subsection, we prove a few claims about the shortcutting algorithm. First, we show adding the shortcut edges does not decrease the distances in $G_{\phi}$. 

    \begin{claim}\label{cl:distances-dont-decrease}
        Let $G'$ denote the graph after our shortcutting procedure. Then $d_{G'}(s,t)=d_{G_{\phi}}(s,t)=d_G(s,t)+\phi(s)-\phi(t)$ for all $s,t\in V$.
    \end{claim}
    \begin{proof}
        Since we only add edges to the graph, it suffices to show that the added edges do not decrease distances in $G_{\phi}$. The weights of shortcut edges \ref{item:shortcut-step-4}--\ref{item:shortcut-step-5} correspond to weights of paths in the graph, so they do not decrease the distances. For other edges added to the Steiner vertex $\tilde r$, we have the following properties:
    	\begin{itemize}
    		\item For every shortcut edge to $\tilde r$ in step \ref{item:shortcut-step-3} and \ref{item:shortcut-step-6}, with the exception of edge $(r',\tilde r)$ added in \ref{item:shortcut-step-6}, for any such edge $(u,\tilde r)$, there is a corresponding path $P$ from $u$ to $r$ such that $\omega(u,\tilde r)=\omega(P)-\Delta_r'$.
    		\item For every shortcut edge from $\tilde r$ in step \ref{item:shortcut-step-2} and \ref{item:shortcut-step-6}, with the exception of edge $(\tilde r,r)$ added in \ref{item:shortcut-step-6}, for any such edge $(\tilde r,w)$, there is a corresponding path $P'$ from $r$ to $w$ such that $\omega(\tilde r,w)=\omega(P')+\Delta_r'$.
    	\end{itemize}
    	We now show that any two-edge subpath that uses $\tilde r$ as the middle vertex does not decrease the distances.
    	\begin{itemize}
    		\item For a path $u\to \tilde r\to w$ where $u\ne r',w\ne r$, we can concatenate path $P$ and $P'$ above to get a path in the previous graph from $u$ to $w$ that has the same weight. Therefore, this subpath never decreases the distances.
    		\item For a path $u\to \tilde r\to r$ where $u\ne r'$, since $\omega(\tilde r,r)=\Delta_r'$, the total weight of this path equals $\omega(P)$, which also begins at $u$ and ends at $r$. Therefore, this subpath never decreases the distances.
    		\item For a path $r'\to \tilde r\to w$ where $w\ne r$, since $\omega(r',\tilde r)=-\omega(r,r')-\Delta_r'$, the weight of this path equals $\omega(P')-\omega(r,r')$. Because we assumed in invariant \ref{item:negative-edge-invariant-1} that negative vertex $r$ only has this outgoing edge $(r,r')$, path $P'$ also starts with this edge, so there is a path of weight $\omega(P')-\omega(r,r')$ from $r'$ to $w$ in the original graph. Therefore, this subpath never decreases the distances.
    		\item Finally, we consider the path $r'\to \tilde r\to r$. Its weight equals $-\omega(r,r')$, so it does not form a negative cycle with the original negative edge $(r,r')$. Since we assume invariant \ref{item:negative-edge-invariant-1} holds, $r$ does not have other outgoing edges, and $r'$ does not have other incoming edges, so this path cannot be extended further.\qedhere
    	\end{itemize}
    \end{proof}

    Next, we show that invariants \ref{item:negative-edge-invariant-1}--\ref{item:steiner-invariant-3} are maintained by the shortcutting procedure. 

    \begin{claim}
        Invariants \ref{item:negative-edge-invariant-1}--\ref{item:steiner-invariant-3} hold for $G'$. 
    \end{claim}
    \begin{proof}
        In addition, we prove that the invariants are still satisfied after the shortcutting step. We consider new edges added in this iteration. Shortcut edges \ref{item:shortcut-step-2}-\ref{item:shortcut-step-3} are managed by the data structure in \Cref{lem:maintain-invariant}, which ensures invariants for them. The remaining shortcut edges are \ref{item:shortcut-step-4}-\ref{item:shortcut-step-5}, which are negative edges added between endpoints of previous negative edges. Since the previous graph satisfies invariant \ref{item:steiner-invariant-3}, these endpoints are original vertices, so these edges also satisfy invariant \ref{item:steiner-invariant-3}.
    \end{proof}
    
	Then, we show that adding these edges shortcuts the paths by a constant factor; the proof is exactly the same as in \cite{li2025shortcutting}.
	
	\begin{lemma}\label[lemma]{lem:shortcut-constant-factor}
		Consider any $a,b\in V$ and any shortest $(a,b)$-path $P$ in $G_{\phi}$ with $h$ negative edges. If $G'$ has shortcut edges \ref{item:shortcut-step-1}--\ref{item:shortcut-step-5}, then there is an $(a,b)$-path in $G'$ with weight $d_{G_{\phi}}(a,b)$ and with at most $h-\lfloor h/3\rfloor$ negative edges.
	\end{lemma}
	\begin{proof}	
		Consider the negative edges of $P$ in order from $a$ to $b$, and create $\lfloor h/3\rfloor$ disjoint groups of three consecutive negative edges. Let $(s,s')$, $(r,r')$, and $(t,t')$ be such a group of three consecutive negative edges. We claim there is an $s\leadsto t'$ path in $G'$ with the same weight as the $s\leadsto t'$ path in $G_{\phi}$, and with one fewer negative edge. Concatenating together each of the 2-hop paths in $G'$ which shortcut each of the $\lfloor h/3\rfloor$ disjoint 3-hop paths in $G_{\phi}$ proves the claim. Our proof will be split into three cases based on $\Delta_r'$.

		\begin{enumerate}
			\item $d^0(s',r)<\Delta_r'$. We have that $s'\in V^{\textup{in}}_r$ so \ref{item:shortcut-step-5} added an edge $(s,r')$ with weight $\omega(s,s')+d^0(s',r)+\omega(r,r')$, which equals the weight of the $s\leadsto r'$ segment of $P$ since the negative edges $(s,s')$ and $(r,r')$ are consecutive along $P$. Therefore, we can shortcut the $s\leadsto r'$ segment of $P$, which has two negative edges, using the single (negative) edge $(s,r')$.
			\item $d^1(r,t)<-\Delta_r'$. We have $t\in V^{\textup{out}}_r$, so \ref{item:shortcut-step-4} added an edge $(r,t')$ of weight $d^1(r,t)+\omega(t,t')$, which equals the weight of the $r\leadsto t'$ segment of $P$ since the negative edges $(r,r')$ and $(t,t')$ are consecutive along $P$. Therefore, we can shortcut the $r\leadsto t'$ segment of $P$, which has two negative edges, using the single (negative) edge $(r,t')$.
			\item $d^0(s',r)\ge \Delta_r'$ and $d^1(r,t)\ge-\Delta_r'$. Consider the $s'\leadsto r$ segment of $P$ and let $u$ be the last vertex on this segment with $u\not\in V_r^{\textup{in}}\cup \{r\}$. By invariant \ref{item:negative-edge-invariant-1}, the negative edges $(s,s')$ and $(r,r')$ do not share an endpoint, so $s'\neq r$, and thus $u$ always exists. There is an edge from $u$ to a vertex $u'\in V_r^{\textup{in}}\cup \{r\}$, so \ref{item:shortcut-step-3} added a shortcut edge $(u,\tilde r)$ of weight $\omega(u,u')+d^0(u',r)-\Delta_r'=d^0(u,r)-\Delta_r'\ge 0$, where the equality holds because we assume the path is the shortest path, and the last inequality holds since $u\not\in V_r^{\textup{in}}$.
			
			Similarly, consider the $r'\leadsto t$ segment of $P$ and let $v'$ be the first vertex on this segment with $v'\not\in V_r^{\textup{out}}\cup\{r'\}$. Such a vertex exists since invariant \ref{item:negative-edge-invariant-1} ensures that $r'\neq t$. There is an edge to $v'$ from a vertex $v\in V_r^{\textup{out}}\cup\{r'\}$, so \ref{item:shortcut-step-2} added a shortcut edge $(\tilde r,v')$ of weight $d^1(r,v)+\omega(v,v')+\Delta_r'=d^1(r,v')+\Delta_r'\ge 0$, where the equality holds because the path is the shortest path, and the inequality holds since $v'\not\in V_r^{\textup{out}}$.
			
			Concatenating the edges $(u,\tilde{r})$ and $(\tilde{r},v')$ in the previous paragraphs together, we get a path $u\to \tilde{r}\to v'$ of non-negative edges, whose total weight equals $d^0(u,r)+d^1(r,v')$. Therefore, we shortcut the $u\leadsto v'$ segment that has one negative edge with non-negative edges.\qedhere
		\end{enumerate}
	\end{proof}

    We now prove some summation bounds that are important for charging arguments later on:
	
	\begin{observation}\label{obs:three-bounds-on-search-sizes}
		Suppose that the $\tilde O(b)$ term in the definition of $\lambda(v)$ is large enough. Then, the following sums of counts are bounded: 
		\begin{itemize}
			\item $\sum_{r\in N}\sum_{u\in \tilde V_r^{\textup{out}}\cap S_0}\sum_{v\in N^{\textup{out}}(u)\cap S_0}\textup{count}^-(v)\le \tilde{O}(m_0h^2b^3)=\tilde{O}(m_0b^3).$
			\item $\sum_{r\in N}\sum_{v\in \tilde V_r^{\textup{in}}\cap S_0}\sum_{u\in N^{\textup{in}}(v)\cap S_0}\textup{count}^+(u)\le \tilde{O}(m_0h^2b^3)=\tilde{O}(m_0b^3).$
		\end{itemize}
		The volume of searches for edges between two original vertices is also bounded:
		\begin{itemize}
			\item $\sum_{r\in N}\sum_{v\in \tilde V_r^{\textup{out}}\cap S_0}\deg_{G[S_0]}(v)\le \tilde{O}(m_0hb^2)=\tilde{O}(m_0b^2)$.
			\item $\sum_{r\in N}\sum_{v\in \tilde V_r^{\textup{in}}\cap S_0}\deg_{G[S_0]}(v)\le \tilde{O}(m_0hb^2)=\tilde{O}(m_0b^2)$.
		\end{itemize}
	\end{observation}
	\begin{proof}
		We first claim the following bound on $\Lambda$:
		
		\begin{subclaim}
			Suppose that the $\tilde O(b)$ term in the definition of $\lambda(v)$ is large enough. Then, $\Lambda\le\tilde O(m_0hb+nhb^2)$.
		\end{subclaim}
		\begin{subproof}
			Suppose that the $\tilde O(\Lambda hb)$ bound in \Cref{lem:multi-hop-forward-backward-search-weighted} is bounded by $\log^Cn\cdot\Lambda hb$ for a large enough constant $C>0$. We set the $\tilde O(hb)$ term in the definition of $\lambda(v)$ to be $h\log^{C+1}n\cdot b$.
			
			Let $\Lambda_t=\sum_{v\in V}\lambda(v)$ be the total weight of vertices after $t$ iterations of shortcutting. Recall that $\lambda(v)$ has three terms: one is a constant $\frac{m_0hb+nhb^2}{n}$, one based on $\deg_{G[S_0]}(v)$, and one based on $\textup{count}(v)$. We analyze the increase of the latter two terms separately. In particular, let $R_t$ denote the total increase in the $\tilde\Theta(hb)\cdot\deg_{G[S_0]}(v)$ term after $t$ iterations of shortcutting and let $\Lambda'_t=\Lambda_t-R_t-(m_0hb+nhb^2)$. We will prove $\Lambda'_t\le\tilde{O}(m_0hb)$ and $R_t=\tilde{O}(nhb^2)$.

            First, we bound $R_t$.
            Recall that in \ref{item:shortcut-step-4} and \ref{item:shortcut-step-5}, we add new edges between original vertices; this is the only cause for $\deg_{G[S_0]}(\cdot)$ increasing. The total number of edges added at iteration $t$ between the original vertices is at most the total number of original vertices in the forward and backward searches in the previous iteration. By \Cref{lem:multi-hop-forward-backward-search-weighted}, we have
            \begin{align*}
                \sum_{r\in N}\sum_{v\in V_r^{out}\cap S_0}\frac{m_0hb+nhb^2}{n}\le
                \sum_{r\in N}\sum_{v\in V_r^{out}\cap S_0}\lambda(v)\le \tilde O(\Lambda_{t-1} hb). 
            \end{align*}
            By the bound on $\Lambda_{t-1}$ in the previous iteration, we know that the total number of edges added between original vertices is at most $\tilde{O}(nb)$. Since there are $\Theta(\log{n})$ iterations and $R_t$ is scaled up by $\tilde\Theta(hb)$, we have that $R_t$ is upper bounded by $\tilde O(nhb^2)$.
            
            Next, we prove by induction on $t$ that
			\[ \Lambda'_t\le \left(1+\frac1{\log n}\right)^t\cdot h\log^{C+1}n\cdot b\cdot 2m_0 .\]
			Since there are only $t=O(\log n)$ iterations of shortcutting, the claim follows.
			
			For the base case $t=0$, we have
			\[ \Lambda'_0=\sum_{v\in S_0}h\log^{C+1}n\cdot b\cdot\deg_{G[S_0]}(v)=h\log^{C+1}n\cdot b\cdot 2m_0 ,\]
			so assume that $t\ge1$. For a vertex $v\in S_0$, let $\textup{count}^-_t(v)$ and $\textup{count}^+_t(v)$ be the number of times $v$ is in the forward and backward search of some negative vertex on iteration $t$ of shortcutting, respectively. By definition of $\Lambda'_t$, we have
			\[ \Lambda'_t=\Lambda'_{t-1}+\sum_{v\in S_0}\bigg( \sum_{w\in N^{\textup{out}}(v)\cap S_0}\textup{count}^-_t(w)+\sum_{u\in N^{\textup{in}}(v)\cap S_0}\textup{count}^+_t(u) \bigg) .\] 
			We first focus on the $\textup{count}^-_t(w)$ terms. By swapping the order of summation, we have
			\begin{align*}
				\sum_{v\in S_0} \text{ } \sum_{w\in N^{\textup{out}}(v)\cap S_0}\textup{count}^-_t(w)=\sum_{w\in S_0} \text{ } \sum_{v\in S_0: w\in N^{\textup{out}}(v)}\textup{count}^-_t(w)=\sum_{w\in S_0}\deg_{G[S_0]}^-(w)\,\textup{count}^-_t(w).
			\end{align*}
			Each pair of vertices $(w,r)$ satisfying $r\in N$ and $w\in \tilde V_r^{\textup{in}}\cap S_0$ contributes exactly $\deg_{G[S_0]}^-(w)$ to the summation above, so the summation also equals
			\[ \sum_{r\in N}\sum_{w\in\tilde V_r^{\textup{in}}\cap S_0}\deg_{G[S_0]}^-(w) .\]
			By definition of $\lambda(w)$, we have $\deg^-_{G[S_0]}(w)\le\lambda(w)/(h\log^{C+1}n\cdot b)$, so we have the final inequality
			\[ \sum_{v\in S_0}\sum_{w\in N^{\textup{out}}(u)\cap S_0}\textup{count}^-_t(w)\le\frac1{h\log^{C+1}n\cdot b}\sum_{r\in N}\sum_{w\in\tilde V_r^{\textup{in}}\cap S_0}\lambda(w) .\]
			Repeating this argument on the $\textup{count}^+_t(w)$ terms, we conclude that
			\[ \Lambda'_t\le\Lambda'_{t-1}+\frac1{h\log^{C+1}n\cdot b}\sum_{r\in N}\bigg(\sum_{w\in\tilde V_r^{\textup{out}}\cap S_0}\lambda(w)+\sum_{w\in\tilde V_r^{\textup{in}}\cap S_0}\lambda(w)\bigg) .\]
			By \Cref{lem:multi-hop-forward-backward-search-weighted} with the bound $\log^Cn\cdot\Lambda' hb$, we obtain
			\[ \Lambda'_t\le\Lambda'_{t-1}+\frac1{h\log^{C+1}n\cdot b}\cdot\log^Cn\cdot\Lambda'_{t-1} hb=\left(1+\frac1{\log n}\right)\Lambda'_{t-1} ,\]
			completing the induction.
		\end{subproof}

        By the subclaim above, we know $\Lambda\le \tilde O(mhb^2)$ since $n\le \tilde{O}(m_0)$. We now finish the proof of \Cref{obs:three-bounds-on-search-sizes}. Since $u\in S_0$, we have $\lambda(u)\ge \sum_{v\in N^{\textup{out}}(u)\cap S_0}\textup{count}^-(v)$, and together with \Cref{lem:multi-hop-forward-backward-search-weighted}, we obtain \begin{align*}
			\sum_{r\in N}\sum_{u\in \tilde{V}^{\textup{out}}_r\cap S_0}\sum_{v\in N^{\textup{out}}(u)\cap S_0}\textup{count}^-(v)\le \sum_{r\in N}\sum_{u\in \tilde{V}^{\textup{out}}_r\cap S_0}\lambda(u)\le \tilde{O}(\Lambda h b)=\tilde{O}(m_0h^2b^3),
		\end{align*}
		and similarly for the backward searches. Similarly, since $\lambda(u)\ge bh\cdot\deg_{G[S_0]}(u)$, we obtain
		\begin{align*}
			\sum_{r\in N}\sum_{u\in \tilde{V}^{\textup{out}}_r\cap S_0}hb\cdot\deg_{G[S_0]}(u)\le \sum_{r\in N}\sum_{u\in \tilde{V}^{\textup{out}}_r}\lambda(u)\le \tilde{O}(\Lambda hb)=\tilde{O}(m_0h^2b^3).
		\end{align*}
		Rearranging gives the desired bound, and the proof is again similar for backward searches.
	\end{proof}

	\subsection{Number of Edges Added}\label{sec:unfolding}
	
	This section proves the final property: the number of edges added in each iteration is $\tilde O(m_0b^2)$.
	
	\begin{lemma} \label{lem:number-of-edges-added} The total number of edges added in each iteration is upper bounded as follows:
		\begin{itemize}
			\item $\sum_{r\in N} |N^{\textup{out}}(V_r^{\textup{out}})|\le \tilde O(m_0b^2)$.
			\item $\sum_{r\in N} |N^{\textup{in}}(V_r^{\textup{in}})|\le \tilde O(m_0b^2)$.
		\end{itemize}
	\end{lemma}
	
	We first prove the following lemma, which states that any path from a negative vertex $r$ to $v\in S_j$ can be ``unfolded'' into a path from $r$ to $v$ such that the final edge $(u,v)$ in the path lies in $S_{<j}$, where we abuse notation and let $S_{<0} = S_0$. The idea is to iteratively apply the unfolding lemma to charge the vertices $v\in N^{\textup{out}}(V^{\textup{out}}_r)$ to edges between original vertices which are on the boundary of $\tilde{V}_r^{\textup{out}}$. Our weights chosen for the betweenness reduction and searches then guarantee that the number of times these edges are charged is $\tilde O(m_0b^2)$. Charging to edges between original vertices is the key to not blowing up the number of edges through iterative shortcutting and recursion: the number of edges added always depends only on $m_0$ instead of $m$.

	\begin{lemma}[Unfolding Lemma]\label{lem:unfolding-new}
		The following (symmetric) statements hold:
		\begin{itemize}
			\item Suppose that there is an $h$-hop path $P$ in $G(t)$ from $r\in N$ to $u$ and let $(w,u)$ be the last~edge in $P$. Then there exists an $(h+t)$-hop path $P'$ in $G(t)$ from $r$ to $u$ of the same weight as~$P$ such that the final edge $(w',u)$ in $P'$ satisfies $i(w')<i(u)$ if $i(u)>0$ and $i(w')=0$ if $i(u)=0$. 
			\item Suppose that there is an $h$-hop path $P$ in $G(t)$ from $u$ to $r\in N$ and let $(u,w)$ be the first edge in $P$. Then there exists an $(h+t)$-hop path $P'$ in $G(t)$ from $u$ to $r$ of the same weight as $P$ such that the first edge $(u,w')$ in $P'$ satisfies $i(w')<i(u)$ if $i(u)>0$ and $i(w')=0$ if $i(u)=0$. 
		\end{itemize}
	\end{lemma}
	\begin{proof}
		We only prove the first statement; the second one is symmetric. Let $u_0=u$, and label the vertices on path $P$ backwards from $u$ as $w=u_1,u_2,\ldots$. Let ${\ell}$ be the first index such that $i(u_{\ell+1})<i(u_{\ell})$; note that together with invariant \ref{item:steiner-invariant-1} this implies $i(u_{\ell})>i(u_{\ell-1})>\cdots>i(u_1)>i(u_0)$. Consider the edges $(u_{\ell+1},u_{\ell})$ and $(u_{\ell},u_{\ell-1})$; since $i(u_{\ell+1})<i(u_{\ell})$ and $i(u_{\ell-1})<i(u_{\ell})$, we know by invariant \ref{item:steiner-invariant-2} that there is some $1$-hop path $Q$ from $u_{\ell+1}$ to $u_{\ell-1}$ of weight $\omega(u_{\ell+1},u_{\ell})+\omega(u_{\ell},u_{\ell-1})$ such that all vertices in $Q$ are in $S_{<i(u_{\ell})}$. Replace the two edges $(u_{\ell+1},u_{\ell})$ and $(u_{\ell},u_{\ell-1})$ in $P$ with the 1-hop path $Q$ to obtain the updated path $P^{\textup{new}}$.
		
		With our updated path $P^{\textup{new}}$, we start at $u^{\textup{new}}_0=u$ and label the vertices on $P^{\textup{new}}$ backwards from $u$ as $u^{\textup{new}}_1,u^{\textup{new}}_2,\ldots$. 
		Let $\ell'$ be first index such that $i(u^{\textup{new}}_{\ell'+1})<i(u^{\textup{new}}_{\ell'})$; we claim $i(u_{\ell'}^{\textup{new}})<i(u_{\ell})$. To prove the claim, observe that $(u_1,\ldots,u_{\ell-1})=(u^{\textup{new}}_1,\ldots,u^{\textup{new}}_{\ell-1})$ are the same vertices before and after updating $P$. We will show that vertices $u^{\textup{new}}_{\ell},\ldots,u^{\textup{new}}_{\ell'}$ lie in $Q$, and hence are in $S_{<i(u_{\ell})}$. First, $u_{\ell}^{\textup{new}}$ is in $Q$ since $P^{\textup{new}}$ was constructed as the concatenation of $Q$ and $(u_{\ell-1},\ldots,u_1)$. Furthermore, path $Q$ always contains a vertex in $S_0$ since endpoints of negative vertices are in $S_0$ (see \ref{item:steiner-invariant-3}). Thus, the entire subpath $u^{\textup{new}}_{\ell},\ldots,u^{\textup{new}}_{\ell'}$ must lie in $Q$ and we have $i(u^{\textup{new}}_{\ell'})<i(u_{\ell})$.
		
		Given that $i(u^{\textup{new}}_{\ell'})<i(u_{\ell})$, we can iteratively unfold edges $(u^{\textup{new}}_{\ell'+1},u^{\textup{new}}_{\ell'})$ and $(u^{\textup{new}}_{\ell'},u^{\textup{new}}_{\ell'-1})$ into a 1-hop path until eventually $i(u^{\textup{new}}_1)<i(u^{\textup{new}}_0)$ and we can no longer unfold. Then we found our path $P'=P^{\textup{new}}$ and $w'=u^{\textup{new}}_1$ satisfying our desired properties. The total number of unfolding steps is at most $t$ since $i(u_\ell)\le t$ at the beginning and decreases by at least one each iteration. Thus, the number of hops in $P^{\textup{new}}$ increased by at most $t$ throughout the unfolding process.
	\end{proof}
	
	\begin{corollary}\label{lem:unfolding}
		The following (symmetric) statements hold:
		\begin{itemize}
			\item 
			
			Suppose that there is an $h$-hop path $P$ in $G(t)$ from $r\in N$ to $u$ and an edge $(u,v)$ in $G(\tau)$. Then there is a vertex $u'$ such that (1) $i(u')<i(v)$ if $i(v)>0$ and $i(u')=0$ if $i(v)=0$ and (2) there exists an $h+O(\tau^3)$-hop path $P'$ in $G(t)$ from $r$ to $u'$ and edge $(u',v)$ in $G(\tau)$ satisfying $\omega_t(P')\le\omega_t(P)$ and $\omega_{\tau}(u',v)\ge\omega_{\tau}(u,v)$. 
			
			\item Suppose that there is an $h$-hop path $P$ in $G(t)$ from $v$ to $r\in N$ and an edge $(u,v)$ in $G(\tau)$. Then there is a vertex $v'$ such that (1) $i(v')<i(u)$ if $i(u)>0$ and $i(v')=0$ if $i(u)=0$ and (2) there exists an $h+O(\tau^3)$-hop path $P'$ in $G(t)$ from $v'$ to $r$ and edge $(u,v')$ in $G(\tau)$ satisfying $\omega_t(P')\le \omega_t(P)$ and $\omega_{\tau}(u,v')\ge\omega_{\tau}(u,v)$. 
		\end{itemize}
	\end{corollary}
	\begin{proof}
		Like before, we prove only the first statement. If $i(u)<i(v)$, we can just return the input path $P$ and original edge $(u,v)$. Otherwise, we will show that there is a vertex $u'$ such that $i(u')<i(u)$, an $h+O(\tau^2)$-hop path $P'$ in $G(t)$ from $r$ to $u'$, and an edge $(u',v)$ in $G(\tau)$ such that $\omega(P)+\omega(u,v)=\omega(P')+\omega(u',v)$ and $\omega(u',v)\ge\omega(u,v)$. Repeatedly applying this until $i(u')<i(v)$ takes at most $t$ iterations, and implies that there is an $h+O(\tau^3)$-hop path $P'$ from $r$ to $v$ such that the final edge $(u',v)$ satisfies $u'\in S_{<i}$ and $\omega(u',v)\ge\omega(u,v)$, as desired.

		To build $P'$, we first apply \Cref{lem:unfolding-new} to obtain an $(h+t)$-hop path $Q$ in $G(t)$ from $r$ to $u$ of the same weight as $P$ such that the final edge $(w',u)$ in $Q$ satisfies $i(w')\le i(u)$. If $i(w')=i(u)$ then $i(w')=i(u)=i(v)=0$ and we're done. Otherwise, we have edges $(w',u)$ and $(u,v)$ which satisfy $i(w')<i(u)$ and $i(v)<i(u)$. We wish to unfold them into a path between $w'$ and $v$. 
        
		By applying \ref{item:steiner-invariant-2} to $(w',u)$ in $G(t)$, there is a 0-hop path $R_1$ of weight $\omega_t(w',u)+\Delta_{r(u)}(t)$ in $G(t)$ from $w'$ to $r(u)$. By applying \ref{item:steiner-invariant-2} to $(u,v)$ in $G(\tau)$, there is a 1-hop path $R_2$ of weight $\omega_{\tau}(R_2)=\omega_{\tau}(u,v)-\Delta_{r(u)}$ in $G(\tau)$ from $r(u)$ to $v$ such that the final edge $(u',v)$ satisfies $\omega_{\tau}(u',v)\ge\omega_{\tau}(u,v)$, as desired. By \Cref{lem:unfold-to-original-graph}, this means there is a $O(\tau^3)$-hop path $R_2'$ in $G(t)$ from $r(u)$ to $u'$ of weight $\omega_t(R_2')\le-\Delta_{r(u)}$. If we define $P'$ as the concatenation $Q[r,w']\circ R_1\circ R_2'$ of paths in $G(t)$, the total weight is $\omega_t(P')\le \omega_t(Q)$. 
	\end{proof}

	Using the unfolding lemma (and its corollary), we bound the number of edges on the boundary of the forward searches; the proof for backward searches follows symmetrically. Consider some neighbor $v\in V\setminus N$ of a vertex $u\in V_r^{\textup{out}}$ in the search; our goal is to bound the number of such $v$.

	Let $r(u)$ and $r(v)$ be the negative vertices from which searches were conducted when Steiner vertices $u$ and $v$ were added. Let ${V}_{r(u)}^{\textup{out}},{V}_{r(u)}^{\textup{in}}$ and $\tilde{V}_{r(u)}^{\textup{out}},\tilde{V}_{r(u)}^{\textup{in}}$ respectively be the $1$-hop and $h$-hop forward and backward searches from $r(u)$ at the iteration when $u$ was added, and similarly for $r(v)$. Let $\Delta_{r(u)}$ denote the value which was used to define $\tilde{V}_{r(u)}^{\textup{out}},\tilde{V}_{r(u)}^{\textup{in}}$, and similarly for $\Delta_{r(v)}$.

	Our goal is to prove that there is an edge $(\tilde u,\tilde v)\in E$ with $\tilde u\in \tilde{V}_{r(u)}^{\textup{out}}\cap S_0$ and $\tilde v\in \tilde{V}_{r(v)}^{\textup{in}}\cap S_0$; we then charge vertex $v$ to the edge $(\tilde{u},\tilde{v})$. If some vertex $v$ is charged to edge $(\tilde u, \tilde v)$, we know that $\tilde{v}\in \tilde{V}_{r(v)}^{\textup{in}}$, so each edge $(\tilde{u},\tilde{v})$ is charged at most $\textup{count}^-(\tilde{v})$ times (i.e., the number of times $\tilde{v}$ is in the backward search from some negative vertex $r$ in previous iterations of shortcutting). The total number of charges can be upper bounded by $\sum_{\tilde{u}\in \tilde{V}_r^{\textup{out}}}\sum_{\tilde{v}\in N^{\textup{out}}(\tilde{u})}\textup{count}^-(\tilde{v})$. Summing over all searches $r\in N$, the total number of edges added is at most $O(m_0\cdot b^2)$ by \Cref{obs:three-bounds-on-search-sizes}.  
	
	
	\begin{claim}
		For each edge $(u,v)\in E$ where $u\in V_r^{\textup{out}}$ and $v\in V\setminus N$, there is an edge $(\tilde u, \tilde v)\in E$ such that $\tilde u\in \tilde{V}_{r(u)}^{\textup{out}}\cap S_0$ and $\tilde v\in \tilde{V}_{r(v)}^{\textup{in}}\cap S_0$.
	\end{claim}
	\begin{proof}
		If $i(u)=0$ and $i(v)=0$, we can set $\tilde{u}=u$ and $\tilde{v}=v$. Otherwise, assume $i(u)<i(v)$; the case of $i(u)>i(v)$ is symmetric. If $i(u)=0$, we find a path $\tilde{P}_v$ and edge $\tilde v\in N^{\textup{out}}(u)$ such that:
		\begin{itemize}
			\item $\tilde{P}_v$ is a $O(\tau^4)$-hop path in $G(i(v))$ from $\tilde{v}$ to $r(v)$
			\item the weight of path $\tilde{P}_v$ was less than $\Delta_{r(v)}$ when $v$ was added to the graph.
		\end{itemize}
		Setting $\tilde u=u$ gives the edge $(\tilde u,\tilde v)$ with the desired properties by definition. 
		
		If $i(u)\neq 0$, we find paths $\tilde{P}_u$, $\tilde{P}_v$, and edge $(\tilde u,\tilde v)\in E\cap (S_0\times S_0)$ such that:
		\begin{itemize}
			\item $\tilde{P}_u$ is a $O(\tau^4)$-hop path in $G(i(u))$ from $r(u)$ to $\tilde{u}$,
			\item the weight of path $\tilde{P}_u$ was less than $-\Delta_{r(u)}$ when $u$ was added to the graph, 
			\item $\tilde{P}_v$ is a $O(\tau^4)$-hop path in $G(i(v))$ from $\tilde{v}$ to $r(v)$
			\item the weight of path $\tilde{P}_v$ was less than $\Delta_{r(v)}$ when $v$ was added to the graph. 
		\end{itemize}
		These four properties imply that $\tilde u \in \tilde{V}_{r(u)}^{\textup{out}}$ and $\tilde v\in \tilde{V}_{r(v)}^{\textup{in}}$, as desired. We build these paths iteratively. At each iteration $\sigma\in\{1,\ldots, \tau\}$, we find two vertices $u^\sigma$ and $v^\sigma$ along with (i) a $O(\sigma \tau^3)$-hop path $P_u^\sigma$ from $r(u)$ to $u^\sigma$ and (ii) a $O(\sigma \tau^3)$-hop path $P_v^{\sigma}$ to $v^{\sigma}$ to $r(v)$.
		
		Our iterative process for defining the two paths is as follows:
		\begin{enumerate}
			\item If $\sigma=1$, we first initialize $v^{0}$ and $P_v^{0}$; otherwise, skip this paragraph. Since there is an edge from $u\in S_{i(u)}$ to $v\in S_{i(v)}$ and $i(u)<i(v)$, we know that $u$ was a neighbor of some vertex $v^{0}$ in the backward search $V_{r(v)}^{\textup{in}}$. This means that there is a 0-hop path $P^0_v$ in $G(i(v))$ from $v^{0}$ to $r(v)$ of weight $\omega(P^0_v)<\Delta_{r(v)}$ such that $\omega(u,v^{0})+\omega(P_v^0)\ge\Delta_{r(v)}$. 
			
			Apply \Cref{lem:unfolding} to path $P^{\sigma-1}_v$ and edge $(u^{\sigma-1},v^{\sigma-1})$ to obtain a vertex, denoted $v^{\sigma}$, along with path $P^{\sigma}_v$ in $G(i(v))$ from $v^{\sigma}$ to $r(v)$ and edge $(u^{\sigma-1},v^{\sigma})$. We have the following properties:
			\begin{itemize}
				\item $\omega_{i}(P^{\sigma}_v)\le \omega_{i}(P^{\sigma-1}_v)<\Delta_{r(v)}$,
				\item $\omega_{\tau}(u^{\sigma-1},v^{\sigma})\ge \omega_{\tau}(u^{\sigma-1},v^{\sigma-1})$, and
				\item $i(v^{\sigma})<i(u^{\sigma-1})$ if $i(u^{\sigma-1})>0$; otherwise $i(v^{\sigma})=0$ as well.
			\end{itemize}
			
			{If $i(u^{\sigma-1})=0$, then we have found our desired $(\tilde u,\tilde v)=(u^{\sigma-1},v^{\sigma})$ and we can terminate.}
			\item If $\sigma=1$, we first initialize $u^0$ and $P_u^0$; otherwise, skip this paragraph. Since there is an edge from $u\in S_{i(u)}$ to $v^{0}\in S_{i(v^0)}$ for some $i(v^0)<i(u)$, we know that $v^{0}$ was a neighbor of some vertex $u^0$ in the forward search $V^{\textup{out}}_{r(u)}$. This means that there is a $1$-hop path $P_u^0$ from $r(u)$ to $u^0$ of weight $\omega(P_u^0)<-\Delta_{r(u)}$ such that $\omega(P_u^0)+\omega(u^0,v^{1})\ge -\Delta_{r(u)}$. 
			
			Apply \Cref{lem:unfolding} to path $P^{\sigma-1}_u$ and edge $(u^{\sigma-1},v^{\sigma})$ to obtain a vertex, denoted $u^{\sigma}$, along with path $P^{\sigma}_u$ in $G(i(u))$ from $r(u)$ to $u^{\sigma}$ and edge $(u^{\sigma},v^{\sigma})$. We have the following properties:
			\begin{itemize}
				\item $\omega_i(P^{\sigma}_u)\le \omega_i(P^{\sigma-1}_u)<-\Delta_{r(u)}$,
				\item $\omega_{\tau}(u^{\sigma},v^{\sigma})\ge\omega_{\tau}(u^{\sigma-1},v^{\sigma})$,
				\item $i(u^{\sigma})<i(v^{\sigma})$ if $i(v^{\sigma})>0$; otherwise $i(u^{\sigma})=0$ as well.
			\end{itemize}
			
			{If $i(v^{\sigma})=0$, then we have found our desired $(\tilde u,\tilde v)$ and we can terminate.}
		\end{enumerate}
		In each iteration, $i(u^{\sigma})$ and $i(v^{\sigma})$ both decrease by at least $1$ and the algorithm terminates when $i(u^{\sigma})=0$ and $i(v^{\sigma})=0$. Since $i(u^0)\le \tau$ and $i(v^0)\le \tau$ by definition, the algorithm terminates in at most $\tau$ iterations. Thus, we found $(\tilde u,\tilde v)$ with the desired properties, completing the proof.
	\end{proof}
	
	\subsection{Maintaining Invariants}
	\label{sec:maintain-invariant}
	
	
	In this section, we show how to maintain the invariants used by the algorithm after reweighting happens, by proving the formal version of \Cref{lem:maintain-invariant}.
	
	
	\begin{lemma}[Formal version of \Cref{lem:maintain-invariant}]\label{lem:maintain-invariant-detailed}
		Let $G=(V,E)$ be a directed graph with no Steiner vertices. We can efficiently maintain a data structure $\textsc{MaintainInvariant}(G)$ that maintains shortcut edges in $G$, with the following operations:
		\begin{enumerate}
			\item $\textsc{Insert}(\vec{\Delta})$: Create a new set of Steiner vertices $S_{t}$ defined by $\vec{\Delta}$. The data structure will maintain non-negative shortcut edges to and from these Steiner vertices (as per \ref{item:shortcut-step-1}--\ref{item:shortcut-step-6}), so that the following holds. For every original vertex $r \in N$, define
			\begin{align*}
				V_r^{\textup{in}}(t)&=\{r\}\cup\{x\in V: d^0(x,r)<\Delta_r\},\\
				V_r^{\textup{out}}(t)&=\{r'\}\cup\{x\in V: d^1(r,x)<-\Delta_r\},
			\end{align*}
			using the current weight. For every vertex $v\in V_r^{\textup{out}}(t)$ and every $w\in N^{\textup{out}}(v)$, if $d^1(r,v)+\omega(v,w)\ge -\Delta_r$ and $d^1(r,v)+\omega(v,w)$ is the actual shortest path distance from $r$ to $w$,\footnote{Throughout this sub-section, we assume a canonical choice of shortest paths. This can be enforced, for example, by appending a random integer tag to each edge weight (which is used for tie breaking lexicographically), and then applying the isolation lemma. The details are omitted in this subsection for the sake of readability.} 
            then after any reweighting, there is a non-negative path of weight $d^1(r,v)+\omega(v,w)+\Delta_r(t)$ from Steiner vertex $\tilde r_t$ to $w$, where $\Delta_r(t)$ is the current relative potential of $\tilde r_t$ to $r$ (and $d^1(r,v)$ and $\omega(v,w)$ are with respect to the re-weighted graph).
			Similarly, for every vertex $v\in V_r^{\textup{in}}(t)$ and every $u\in N^{\textup{out}}(v)$, if $\omega(u,v)+d^0(v,r)\ge \Delta_r$ and $\omega(u,v)+d^0(v,r)$ is the actual shortest path distance from $u$ to $v$, then after any reweighting, there is a non-negative path of weight $\omega(u,v)+d^0(v,r)-\Delta_r(t)$ from $u$ to Steiner vertex $\tilde r_t$.
			\item $\textsc{Reweight}(\phi)$: Reweight the graph via a valid potential function $\phi$. The data structure will maintain Steiner edges so that they still satisfy invariants \
			\ref{item:steiner-invariant-1}--\ref{item:steiner-invariant-3} and the shortcut property above.
		\end{enumerate}
	\end{lemma}
	
	The non-trivial part of this data structure is to maintain the edge weight property of invariant \ref{item:steiner-invariant-2} about the first and the last edges of the unfolded path. Recall the invariant requires that, informally, for any non-negative two-edge path $u\to \tilde v\to w$, we can unfold it to a path from $u$ to $w$, where the weights of the first and the last edge on the path does not decrease after unfolding. This property holds if no reweighting happens, because in the construction, the vertex before $w$ in the unfolded path is in $V_r^{\textup{out}}(t)$, so the distance from $r$ to it on the path is less than $-\Delta_r$, which implies that $\omega(\tilde r,w)$ is less than the weight of the last edge. Similarly, the vertex after $u$ in the unfolded path is in $V_r^{\textup{in}}(t)$, so the distance from it to $r$ on the path is less than $\Delta_r$. The issue is that, after reweighting happens, these properties may break since $V_r^{\textup{out}}(t)$ and $V_r^{\textup{in}}(t)$ are defined on the original weights, as illustrated in \Cref{fig:maintain-invariant-example}.
	
	\begin{figure}[ht]
		\centering
		\includegraphics[width=\textwidth]{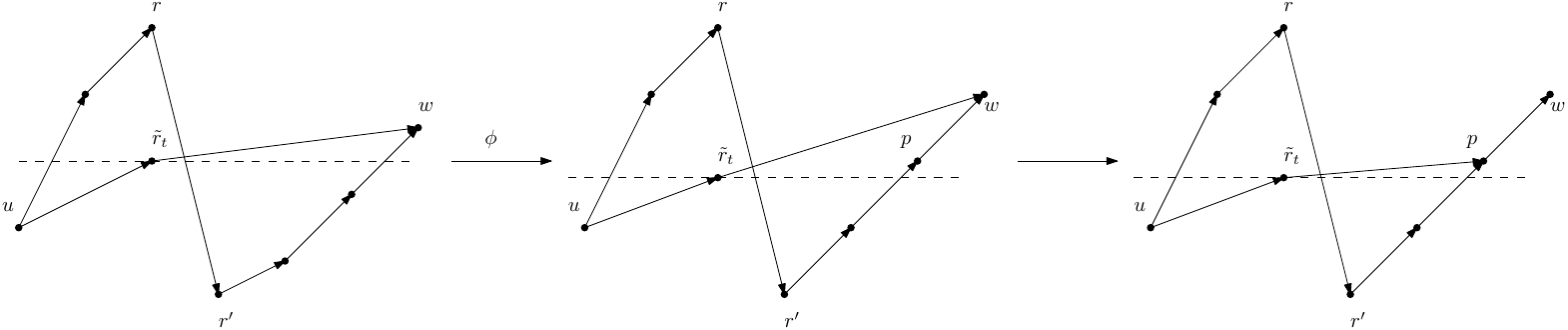}
		\caption{Illustration of the main issue and the intuition of our solution to maintain invariants after reweightings. The first two parts show that invariant \ref{item:steiner-invariant-2} may be violated after reweighting because the distance from $r$ to a vertex $p\in V_r^{\textup{out}}(t)$ may exceed $-\Delta_r(t)$. The final part shows our solution: If such case happens, we replace the shortcut edge from $\tilde r_t$ to $w$ by an edge from $\tilde r_t$ to $p$.}
		\label{fig:maintain-invariant-example}
	\end{figure}

	
	The intuition to fix this, in the bad case illustrated by this figure, where we can find another vertex $p$ on the path from $r'$ to $w$ with $d^1(r,p)\ge -\Delta_r$, is to use a non-negative edge of weight $d^1(r,p)+\Delta_r$ from $\tilde r_t$ to $p$ to replace the starting shortcut edge to $w$. The non-negative path $\tilde r_t\to p\leadsto w$ now acts like the starting shortcut edge $(\tilde r_t,w)$, so the progress of shortcutting is kept. Moreover, if we choose the first such $p$ on the path, then the distance from $r$ to the vertex before $p$ is less than $-\Delta_r$, so invariant \ref{item:steiner-invariant-2} holds. This requires looking at the unfolded paths from $r$ to vertices in the forward or backward searches, so our data structure stores these paths explicitly, as a shortest path tree to and from $r$. As we update the graph and its shortcut edges, we update this tree accordingly.
	
	We now formally state our data structure. For a Steiner copy $\tilde r_t$ of a vertex $r$ at level $t$, we maintain a scalar $\Delta_r(t)$ that stores the relative potential difference between it and the original vertex $r$, with two directed trees. These trees are meant to be shortest path trees to and from $r$. First, we maintain a directed tree $T_r^{\textup{out}}(t)$ from vertex $r$, with the following properties.
	\begin{itemize}
		\item All edges in $T_r^{\textup{out}}(t)$ are existing edges in the subgraph $G[S_{<t}]$. Apart from the edge $(r,r')$, all other edges have non-negative weights.
		\item Apart from $r$, $r'$, and the leaf nodes of this tree, the distance from $r$ to other vertices on this tree is less than $-\Delta_r(t)$.
	\end{itemize}
	Second, we also maintain a directed tree $T_r^{\textup{in}}(t)$ to vertex $r$, with the following properties.
	\begin{itemize}
		\item All edges in $T_r^{\textup{in}}(t)$ are existing edges in the subgraph $S_{<t}$, and they all have non-negative weights.
		\item Apart from $r$ and the leaf nodes of this tree, the distance from other vertices on this tree to $r$ is less than $\Delta_r(t)$.
	\end{itemize}
	
	With these trees and $\Delta_r(t)$, we define the shortcut edges between $\tilde r_t$ and $S_{<t}$ as follows.
	\begin{itemize}
		\item For every leaf node $w$ in $T_r^{\textup{out}}(t)$ with $d_{T_r^{\textup{out}}(t)}(r,w)\ge -\Delta_r$, add the edge $(\tilde r_t,w)$ of weight $d_{T_r^{\textup{out}}(t)}(r,w)+\Delta_r$.
		\item For every leaf node $u$ in $T_r^{\textup{in}}(t)$ with $d_{T_r^{\textup{in}}(t)}(u,r)\ge \Delta_r$, add the edge $(u,\tilde r_t)$ of weight $d_{T_r^{\textup{out}}(t)}(u,r)-\Delta_r$.
	\end{itemize}
	
	For insertion operations, we use \Cref{lem:forward-backward-search-weighted} to compute the forward search $V_r^{\textup{out}}(t)$ and backward search $V_r^{\textup{in}}(t)$ for every original vertex $r$, along with the shortest path trees. We then enumerate the neighbors of forward and backward reaches, and include them into the shortest path tree as leaf nodes. For every $w\in N^{\textup{out}}(V_r^{\textup{out}}(t))$ with $w\not\in V_r^{\textup{out}}(t)$, we find $v\in V_r^{\textup{out}}(t)$ that minimizes $d^1(r,v)+\omega(v,w)$, and add edge $(v,w)$ into $T_r^{\textup{out}}(t)$. 
    By our definition of the shortcut edges, this gives a shortcut edge $(\tilde r_t,w)$ of weight $d^1(r,v)+\omega(v,w)+\Delta_r$.
	Similarly, for every $u\in N^{\textup{in}}(V_r^{\textup{in}}(t))$ with $u\not\in V_r^{\textup{in}}(t)$, we find $v\in V_r^{\textup{in}}(t)$ that minimizes $\omega(u,v)+d^0(v,r)$, and add edge $(u,v)$ into $T_r^{\textup{in}}(t)$. 
    This gives a shortcut edge $(u,\tilde r_t)$ of weight $\omega(u,v)+d^0(v,r)-\Delta_r$. Therefore, the properties hold after this insertion.
	
	We now state the process of maintaining trees after reweighting. After reweightings $\phi$, we first update each $\Delta_r(t)\leftarrow \Delta_r(t)+\phi(\tilde{r}_t)-\phi(r)$. To update the trees $T_r^{\textup{out}}(t)$ and $T_r^{\textup{in}}(t)$ with the corresponding shortcut edges, for $t\in [\tau]$ in increasing order, we do the following.
	\begin{itemize}
		\item Update the trees $T_r^{\textup{out}}(t)$ and $T_r^{\textup{in}}(t)$ as follows: If a node $p$ other than $r,r'$ in $T_r^{\textup{out}}(t)$ has $d_{T_r^{\textup{out}}(t)}(r,p)\ge -\Delta_r(t)$, we make it as a leaf node by removing other nodes in its subtree. Symmetrically, if a node $p$ other than $r$ in $T_r^{\textup{in}}(t)$ has $d_{T_r^{\textup{in}}(t)}(p,r)\ge \Delta_r(t)$, we similarly make it as a leaf node by removing other nodes in its subtree.
		\item Remove previous shortcut edges from $\tilde r_t$ to $S_{<t}$, and replace them with the shortcut edges defined by the current tree. That is, we add shortcut edges from the leaf nodes of the current tree $T_r^{\textup{in}}(t)$ to $\tilde r_t$, and from $\tilde r_t$ to the leaf nodes of the current tree $T_r^{\textup{out}}(t)$.
		\item During the second step, we may remove a shortcut edge between $\tilde r_t$ and a vertex $w$. This happens when in the first step, some ancestor $p$ of $w$ was marked as a leaf node, so $w$ was removed from the tree. This shortcut edge may be used by trees at higher Steiner levels, so we need to replace the edge in these trees by a non-negative path of the same weight. If $w$ was in $T_r^{\textup{out}}(t)$, then this path is $(\tilde{r}_t,p)\circ T_r^{\textup{out}}(t)[p,w]$. If $w$ was in $T_r^{\textup{in}}(t)$, then this path is $T_r^{\textup{out}}(t)[w,p]\circ (p,\tilde r_t)$.
		
		If this edge is a boundary edge that points to a leaf node outside the search $V$, then we only add the part of this path that is in the search by the distance bound $\Delta$. Moreover, if this path intersects with the original tree during the replacement, then we only keep the shorter path to the intersection vertex. 
        The subtree of the non-shortest path cannot be part of the real shortest paths, so we remove that subtree. This maintains the tree structure after replacements.
	\end{itemize}
	
	We prove that all invariants are maintained after these steps, as desired. Invariant \ref{item:steiner-invariant-3} is about negative edges, so it holds since we only add non-negative shortcut edges.
	For invariant \ref{item:steiner-invariant-1}, we can inductively prove that, since the trees of Steiner vertices at level $t$ are in $S_{\le t}$, the replacement path for an edge in $S_{\le t}$ found in the third step are also in $S_{\le t}$, and thus the trees of Steiner vertices at level $t$ are still in $S_{\le t}$ after replacements. This ensures that the trees for Steiner vertices at level $t$ only add edges between them and $S_{<t}$, which implies \ref{item:steiner-invariant-1}.
	
	We now prove invariant \ref{item:steiner-invariant-2}. First, we show how to unfold shortcut edges in the \emph{current} graph $G'(\tau)$. For a shortcut edge $(\tilde r_t,w)$ where $w\in S_{<t}$, $w$ is a leaf node in the tree $T_r^{\textup{out}}(t)$, and the weight of the path from $r$ to $w$ in the tree equals $\omega(\tilde r_t,w)-\Delta_r(t)$. 
	Moreover, let $p$ be the parent node of $w$ in this tree; we will show that $\omega(p,w)\ge \omega(r_t,w)$. It suffices to prove that $d_{T_r^{\textup{out}}(t)}(r,p)\le -\Delta_r(t)$, since it implies that \[\omega(p,w)=d_{T_r^{\textup{out}}(t)}(r,w)-d_{T_r^{\textup{out}}(t)}(r,p)\ge d_{T_r^{\textup{out}}(t)}(r,w)+\Delta_r(t)=\omega(r_t,w).\] If $p\ne r'$, then this holds since in the first step above, we mark vertices other than $r'$ that violates this condition as leaf nodes of the tree. If $p=r'$, this simplifies to $\Delta_r(t)\le -\omega(r,r')$, which is ensured by the shortcut step \ref{item:shortcut-step-6} as follows.
	\begin{claim}
		At any iteration, we have $\Delta_r(t)\in[0,-\omega(r,r')]$.
	\end{claim}
	\begin{proof}
		Let $(\Delta_r(t))^{\textup{old}}$ denote the value of $\Delta_r(t)$ at the iteration when the Steiner vertex $\tilde r_t$ is added, and let $(\omega(r,r'))^{\textup{old}}$ denote the negative edge weight at that iteration. Recall that shortcut step \ref{item:shortcut-step-6} adds a non-negative edge $(\tilde r_t,r)$ of weight $(\Delta_r(t))^{\textup{old}}$, and a non-negative edge $(r',\tilde r_t)$ of weight $-(\omega(r,r'))^{\textup{old}}-(\Delta_r(t))^{\textup{old}}$. Since these are non-negative edges, the subsequent reweighting $\phi_{(t,\tau]}$ satisfies $\phi_{(t,\tau]}(r)\le \phi_{(t,\tau]}(\tilde r_t)+(\Delta_r(t))^{\textup{old}}$ for the first edge $(\tilde r_t,r)$. This gives
		\[\Delta_r(t)=(\Delta_r(t))^{\textup{old}}+\phi_{(t,\tau]}(\tilde r_t)-\phi_{(t,\tau]}(r)\ge 0.\]
		Similarly, the second edge $(r',\tilde r_t)$ gives $\phi_{(t,\tau]}(\tilde r_t)\le \phi_{(t,\tau]}(r')-(\omega(r,r'))^{\textup{old}}-(\Delta_r(t))^{\textup{old}}$, so
		\[\Delta_r(t)=(\Delta_r(t))^{\textup{old}}+\phi_{(t,\tau]}(\tilde r_t)-\phi_{(t,\tau]}(r)\le -(\omega(r,r'))^{\textup{old}}+\phi_{(t,\tau]}(r')-\phi_{(t,\tau]}(r)=-\omega(r,r').\qedhere\]
	\end{proof}
	Similarly, for a shortcut edge $(u,\tilde r_t)$ where $u\in S_{<t}$, $u$ is a leaf node in the tree $T_r^{\textup{in}}(t)$, and $\omega(u,\tilde r)=d_{T_r^{\textup{in}}(t)}(u,r)-\Delta_r(t)$. To prove that $\omega(u,p)\ge \omega(u,\tilde r)$, it suffices to prove that $d_{T_r^{\textup{in}}(t)}(p,r)\le\Delta_r(t)$. Let $p$ be the parent node of $w$ in this tree. If $p\ne r$, then this similarly follows from the first step of maintaining the shortcut edges. Otherwise, $p=r$ and this simplifies to $\Delta_r(t)\ge 0$, which follows from the claim above.
	We unfold the shortcut edges to and from $\tilde r_t$ to the respective paths in the current graph, and the weight of the first and last edges on the unfolded path never decreases.
	
	The non-trivial part is to fully unfold the path from $\tilde r_t$ to $p$ using edges that existed when $\tilde r_t$ is added. We state our goal as the following lemma, which implies the rest part of invariant \ref{item:steiner-invariant-2}.
	
	\begin{lemma}\label{lem:unfold-to-original-graph}

		For a vertex $p\in T_r^{\textup{out}}(t)$ with $d_{T_r^{\textup{out}}(t)}(r,p)<-\Delta_r(t)$, we can find a $O(\tau^3)$-hop path from $r$ to $p$ in $G(t)$ of weight $d_{T_r^{\textup{out}}(t)}(r,p)-\phi_{(t,\tau]}(r)+\phi_{(t,\tau]}(p)$, which is the reweighted version of current path at iteration $t$. Furthermore, the weight of the path is less than $-\Delta_r(t)+\phi_{(t,\tau]}(\tilde{r}_t)-\phi_{(t,\tau]}(r)$, which is the value of $-\Delta_r(t)$ when it is originally computed at iteration $t$.

		Symmetrically, for a vertex $p\in T_r^{\textup{in}}(t)$ with $d_{T_r^{\textup{in}}(t)}(p,r)<\Delta_r(t)$, we can find a $O(\tau^3)$-hop path from $p$ to $r$ in the graph at the time when Steiner vertex $\tilde r_t$ is added, whose weight is less than $\Delta_r$, the relative potential of $\tilde r_t$ when it is added.
	\end{lemma}
	
	\begin{proof}
		We consider the case when $\Delta_r(t)$ is not clamped, i.e. $\Delta_r(t)\in (0,-\omega(r,r'))$ and thus $\Delta_r(t)=\Delta'_r(t)$ at iteration $t$. We first show that other cases can be proved by proving this case. First, if $\Delta_r(t)<0$, then the following proof can still be applied for the $T_r^{\textup{out}}(t)$ side, since it actually proves a bound of $-\Delta_r'(t)$, which is larger than $-\Delta_r(t)$. For the $T_r^{\textup{in}}(t)$ part, we can show that this tree only contains one non-leaf node $r$ if $\Delta_r(t)$ was clamped to $0$, so we already have the bound. Recall that in shortcut step \ref{item:shortcut-step-6}, we add bidirectional zero-weight edge between $r$ and $\tilde r_t$ in this case. Therefore, the potential of $r$ and $\tilde r_t$ must be equal in subsequent reweightings, and thus $\Delta_r(t)$ stays at $0$. This implies that every node other than $r$ is a leaf node in this tree. The case when $\Delta_r(t)>-\omega(r,r')$ can be reduced similarly, so it suffices to only consider the case when $\Delta_r(t)$ is not clamped.
		
		We focus on the first case since the two cases are symmetric. 
		We first construct the path from $r$ to $p$, and then prove that its weight is less than $-\Delta_r$.
		
		For construction, we use the following tool, which is based on the Unfolding Lemma (\Cref{lem:unfolding-new}) on the graph $G(t)$.
		\begin{subclaim}\label{subcl:unfold-u-to-ru}
			In graph $G(t)$, suppose that there is a $h$-hop path $P$ from $r\in N$ to $u$, then there exists an $(h+t)$-hop path $P^r$ from $r$ to $r(u)$, such that its weight $\omega(P^r)$ equals $\omega(P)+\Delta_{r(u)}(i(u))$, where the weights and $\Delta$ values are in $G(t)$.

			Symmetrically, in graph $G(t)$, suppose that there is a $h$-hop path $P$ in from $u$ to $r\in N$, then there exists an $(h+t+1)$-hop path $P^r$ from $r(u)$ to $r$, such that its weight $\omega(P^r)$ equals $\omega(P)-\Delta_{r(u)}(i(u))$, where the weights and $\Delta$ values are in $G(t)$.
		\end{subclaim}
		
		\begin{subproof}
			We focus on proving the first case as the second one is symmetric. If $i(u)=0$, then we are done by taking $P^r=P$. Therefore, we assume that $i(u)>0$. In this case, \Cref{lem:unfolding-new} gives a $(h+t)$-hop path $P'$ from $r$ to $u$ that has the same weight as $\omega(P)$, and the last vertex $w'$ before $u$ satisfies $i(w')<i(u)$. Since $i(w')<i(u)$, this shortcut edge $(w',u)$ exists because $w'\in T_{r(u)}^{\textup{in}}(i(u))$. Moreover, $\omega(w',u)$ equals $d_{T_{r(u)}^{\textup{in}}(i(u))}(w',r(u))-\Delta_{r(u)}(i(u))$, so the length of the path from $w'$ to $r(u)$ on this tree equals $\omega(w',u)+\Delta_{r(u)}(i(u))$. By replacing the last edge $(w',u)$ on $P'$ with this path, we get a $(h+t)$-hop path of weight $\omega(P)+\Delta_{r(u)}(i(u))$ from $r$ to $r(u)$.
		\end{subproof}
		
		The construction is done recursively; see \Cref{fig:path-construction} for a visual reference. Before the $k$-th step, we are left with an edge $(u_k,w_k)$ in $G(t)$, such that we have already constructed a path from $r$ to $u_k$ in $G(t)$, and vertex $p$ is on the path from $u_k$ to $w_k$ that is expanded from edge $(u_k,w_k)$ after the data structure maintains the shortcut edges for various rounds. Moreover, this path from $r$ to $u_k$ has weight $d_{T_r^{\textup{out}}(t)}(r,u_k)-\phi_{(t,\tau]}(r)+\phi_{(t,\tau]}(u_k)$. 
		
		Suppose that $i(u_k)\ge i(w_k)$. In this case, the shortcut edge $(u_k,w_k)$ corresponds to the path from $r(u_k)$ to $w_k$ in the shortest path tree $T_{r(u_k)}^{\textup{out}}(i(u_k))$ in $G(t)$. In the current graph $G'(\tau)$, each edge in this path from $r(u_k)$ to $w_k$ may get replaced by a path, and the current path from $u_k$ to $w_k$ starts by a shortcut edge from $u_k$ to a vertex on this path, and then follows the remaining suffix of this path. 
		Let $u_{k+1}$ be the last vertex on the path from $u_k$ to $p$ in $G'(\tau)$ that also occurs on the path from $r(u_k)$ to $w_k$ in $G(t)$, and let $w_{k+1}$ be the next vertex on this path. We construct the path from $r$ to $u_{k+1}$ as follows. First, we apply \Cref{subcl:unfold-u-to-ru} on the previously constructed path from $r$ to $u_k$, to get a path from $r$ to $r(u_k)$. This step adds $O(\tau)$ hops to this path, and the length of this path is increased by the value of $\Delta_{r(u_k)}(i(u_k))$ at iteration $t$. Next, we concatenate this with the path from $r(u_k)$ to $u_{k+1}$ on the shortest path tree at iteration $t$. From the construction of shortcut edges, the length of this path in $G(t)$ equals the length of the path from $u_k$ to $u_{k+1}$ in $G(t)$, subtract the value of $\Delta_{r(u_k)}(i(u_k))$ at iteration $t$. These two $\Delta_{r(u_k)}(i(u_k))$ cancel out, and the total length of this path from $r$ to $u_{k+1}$ at iteration $t$ equals
		\begin{align*}
			&\left(d_{T_r^{\textup{out}}(t)}(r,u_k)-\phi_{(t,\tau]}(r)+\phi_{(t,\tau]}(u_k)\right) + \left(d_{T_r^{\textup{out}}(t)}(u_k,u_{k+1})-\phi_{(t,\tau]}(u_k)+\phi_{(t,\tau]}(u_{k+1})\right)\\
			=& \left(d_{T_r^{\textup{out}}(t)}(r,u_{k+1})-\phi_{(t,\tau]}(r)+\phi_{(t,\tau]}(u_{k+1})\right).
		\end{align*}
		If $p=u_{k+1}$, then we are done. Otherwise, we solve the $(u_{k+1},w_{k+1})$ case recursively, treating $u_{k+1}$ as the new vertex $r$.
		
		If $i(u_k)<i(w_k)$, we similarly consider the path from $u_k$ to $r(w_k)$ in the shortest path tree $T_{r(w_k)}^{\textup{in}}(i(w_k))$. Again, let $u_{k+1}$ be the last vertex on the path from $u_k$ to $p$ in $G'(\tau)$ that also occurs on this path in $G(t)$, and $w_{k+1}$ be the next vertex on this path. In this case, we construct the path from $r$ to $u_{k+1}$ by concatenate the path from $r$ to $u_k$ with the path from $u_k$ to $u_{k+1}$. The length of the latter subpath equals $d_{T_r^{\textup{out}}(t)}(u_k,u_{k+1})-\phi_{(t,\tau]}(u_k)+\phi_{(t,\tau]}(u_{k+1})$, so the total length of the path from $r$ to $u_{k+1}$ also equals $d_{T_r^{\textup{out}}(t)}(r,u_{k+1})-\phi_{(t,\tau]}(r)+\phi_{(t,\tau]}(u_{k+1})$ in this case. See \Cref{fig:path-construction} for an illustration of this construction.
		
		\begin{figure}[ht]
			\centering
			\includegraphics[width=\textwidth]{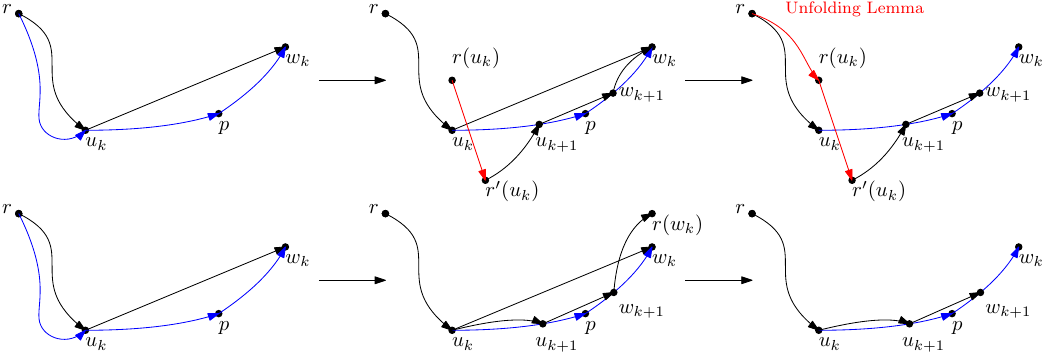}
			\caption{Illustration of the construction step. Black and red edges are non-negative and negative edges in the $G(t)$, while blue edges are edges in $G'(\tau)$. The top figure corresponds to the case when $i(u_k)\ge i(w_k)$, and the shortest path tree that gives edge $(u_k,w_k)$ is rooted at $r(u_k)$. The bottom figure corresponds to the case when $i(u_k)<i(w_k)$.}
			\label{fig:path-construction}
		\end{figure}
		
		In each step, $\max\{i(u_k), i(w_k)\}$ decreases by at least $1$. \Cref{subcl:unfold-u-to-ru} adds at most $O(\tau^2)$ hops in each step, so the total number of hops in the final path is bounded by $O(\tau^3)$.
		
		Next, we prove that the weight of this path is less than the value of $-\Delta_r(t)$ in $G(t)$. We prove this by contradiction. In particular, we show that if the weight of this path is at least $-\Delta_r$, then we can find a parent node of $q\ne r'$ in tree $T_r^{\textup{out}}(t)$ that the distance from $r$ to $q$ on the tree is at least $-\Delta_r$, in $G(t)$. This implies that in $G'(t)$ and later graphs, $p$ will be removed from the tree after the data structure maintains the shortcut edges, which gives the contradiction.
		
		If the distance from $r$ to $p$ is at least $-\Delta_r(t)$, then we can find some step $k$ that in this step, the path from $t$ to $u_k$ has weight at least $-\Delta_r(t)$. We find the smallest such $k$.
		Now we inductively show that, for every $j\le k$, after the data structure performs the reweighting operation when adding $\tilde r_t$, there is a path on the tree $T_r^{\textup{out}}(t)$ from $u_j$ to $u_k$, whose weight equals the weight from $u_j$ to $u_k$ on the unfolded path. Suppose we have such a path from $u_{j+1}$ to $u_k$, we now construct the path from $u_j$ to $u_k$.
		Suppose that $i(u_j)\ge i(w_j)$ in step $j$. In this case, we have a path from $r(u_j)$ to $u_{j+1}$ in the shortest path tree, where only negative edge on this path is the edge from $r(u_j)$. We now concatenate this path with the path from $u_{j+1}$ to $u_k$ constructed by induction. Since the distance from $r$ to $u_j$ is less than $-\Delta_r(t)$, but the distance from $r$ to $u_k$ is at least $-\Delta_r(t)$, we know the distance from $r(u_j)$ to $u_k$ is at least $-\Delta_{r(u_j)}(i(u_j))$. Therefore, we have a non-negative path from $u_j$ to $u_i$ after the data structure maintains the shortcut edges from $u_j$.
		If $i(u_j)<i(w_j)$, then $u_{j+1}$ is on the path from $u_j$ to $w_j$. In this case, we directly concatenate the non-negative subpath from $u_j$ to $u_{j+1}$ to the remaining path. This completes the induction. 
		By repeating this step until $i=0$, we find a path on the tree $T_r^{\textup{out}}(t)$ from $r'$ to $p$, and the weight of edge $(r,r')$ concatenated with this path is at least $-\Delta_r(t)$. Therefore, when the data structure maintains the shortcut edges from $\tilde r_t$, we know that $u_k$ should be a leaf node, and $p$ will be removed from the tree, a contradiction.
		\begin{figure}[ht]
			\centering
			\includegraphics[width=\textwidth]{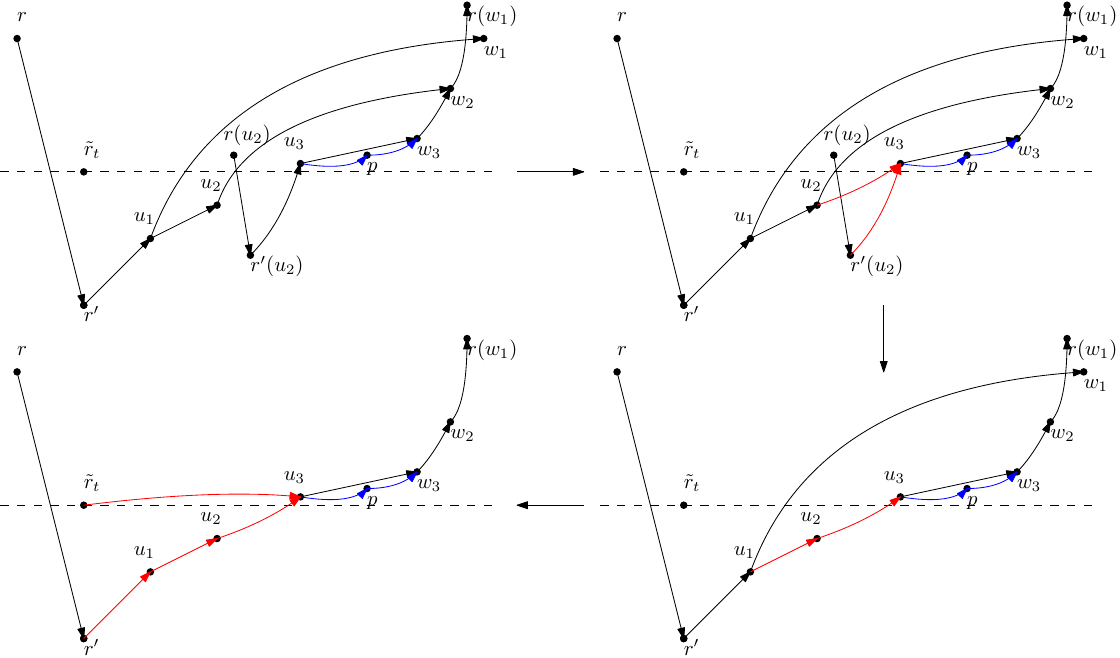}
			\caption{Illustration of the contradiction proof. In each step, we find a non-negative path from $u_j$ to $u_3$. The first step corresponds to the case of $i(u_k)\ge i(w_k)$, and the second step corresponds to the case of $i(u_k)<i(w_k)$. The final step of finding a non-negative path from $\tilde r$ to $u_3$ is same as the $i(u_k)\ge i(w_k)$ case.}
		\end{figure}
	\end{proof}
	
	Finally, we show that \Cref{lem:unfold-to-original-graph} also implies the efficiency of our data structure. In particular, it implies that all non-leaf vertices in the tree at any time are in the original $O(\tau^4)$-hop forward and backward reach $\tilde V_r^{\textup{out}}$ and $\tilde V_r^{\textup{in}}$. In \Cref{lem:number-of-edges-added}, we prove that the total number of neighboring vertices of these reaches is $\tilde O(m_0b^2)$. Therefore, the total size of all trees at one level is also $\tilde O(m_0b^2)$, so the total number of shortcut edges is $\tilde O(m_0\tau b^2)$. A simple charging argument shows that the trees can be updated efficiently. 

    \subsection{Proof of \Cref{thm:one-iteration-shortcutting}}

    The algorithm is described in \Cref{sec:algorithm}. We first apply the recursive betweenness reduction algorithm (\Cref{lem:betweenness-reduction-weighted}), which takes $T(\tilde O(m),k/2^{\sqrt{\log{n}}})+\tilde O(m)$ time. Next, we compute the (weighted) forward and backward searches from each negative vertex (\Cref{lem:multi-hop-forward-backward-search-weighted}, which takes $m^{1+o(1)}$ time since $b=2^{\sqrt{\log{n}}}$. Finally, we add shortcut edges \ref{item:shortcut-step-1}--\ref{item:shortcut-step-6} and update the data structure as in \Cref{lem:maintain-invariant-detailed}, which takes the same $m^{1+o(1)}$ time as the searches. The runtime of \Cref{thm:one-iteration-shortcutting} follows.

    Property 1, which requires shortest path distances are preserved, follows by \Cref{cl:distances-dont-decrease}. Property 2, which requires that the hop-diameter decreases, follows by \Cref{lem:shortcut-constant-factor}. Property 3, which requires that only $O(k)$ new vertices are added to the graph, follows since vertices are only added in shortcut step \ref{item:shortcut-step-1} and \Cref{lem:one-negative-outgoing-edge}, both of which add $O(k)$ vertices. Property 4, which requires that only $m_0^{1+o(1)}$ edges are added to the graph, follows by \Cref{lem:number-of-edges-added}. Property 5, which requires that no negative vertices are introduced, follows from the fact that negative edges are only ever added as out-edges from negative vertices (specifically in \ref{item:shortcut-step-4} and \ref{item:shortcut-step-5}). 
    
	\bibliographystyle{alpha}
	\bibliography{ref}
	
\end{document}